\documentclass{article}
\usepackage{log_2024}						

\usepackage{booktabs}						
\usepackage{multirow}						
\usepackage{amsfonts}						
\usepackage{graphicx}						
\usepackage{duckuments}						

\newcommand{\bs}[1]{\boldsymbol{#1}} 
\usepackage{tabularx}
\usepackage{bbm}
\usepackage{comment}
\usepackage{todonotes}
\usepackage{mathrsfs}
\usepackage{amsthm}
\newtheorem{proposition}{Proposition}

\usepackage{makecell}
\usepackage{multirow}
\usepackage{array}

\usepackage{algorithm}
\usepackage{algpseudocode}
\usepackage{amsmath}
\usepackage{amssymb}
\usepackage{verbatim}
\usepackage{subcaption}
\usepackage{cancel}
\usepackage{ulem}
\usepackage{caption}
\usepackage{sidecap}
\sidecaptionvpos{figure}{c}
\algrenewcommand\algorithmicensure{\textbf{Output:}}
\algrenewcommand\algorithmicrequire{\textbf{Input:}}

\newcommand{\new}[1]{\textcolor{black}{#1}}
\newcommand{\newnew}[1]{\textcolor{black}{#1}}

\usepackage[numbers,compress,sort]{natbib}	

\title[A Spectral Framework for Tracking Communities in Evolving Networks]{A Spectral Framework for Tracking Communities\\ in Evolving Networks}

\author[J. Hume and L. Balzano]{%
Jacob Hume\\
University of Cambridge\thanks{\newnew{Work completed while the author was at Michigan.}}\\
\email{jakehume@umich.edu}\And
Laura Balzano \\
University of Michigan, Ann Arbor\\
\email{girasole@umich.edu}
}

\begin{document}

\maketitle

\begin{abstract}
Discovering and tracking communities in time-varying networks is an important task in network science, motivated by applications in fields ranging from neuroscience to sociology.  In this work, we characterize the celebrated family of spectral methods for static clustering in terms of the low-rank approximation of high-dimensional node embeddings. 
From this perspective, it becomes natural to view the evolving community detection problem as one of subspace tracking on the Grassmann manifold. 
While the resulting optimization problem is nonconvex, we adopt a recently proposed block majorize-minimize Riemannian optimization scheme to learn the Grassmann geodesic which best fits the data. Our framework generalizes any static spectral community detection approach and leads to algorithms achieving favorable performance on synthetic and real temporal networks, including those that are weighted, signed, directed, mixed-membership, multiview, hierarchical, cocommunity-structured, bipartite, or some combination thereof.
We demonstrate how to specifically cast a wide variety of methods into our framework, and demonstrate greatly improved dynamic community detection results in all cases.

\end{abstract}

\section{Introduction}
Consider a multiplex graph whose layers represent snapshots $G_{1},\dots,G_{T}$ of a time-varying network with $d$ nodes. 
The goal is to partition each $G_{i}$ into communities — node groupings of relatively higher intraconnectivity — in a manner that is temporally \new{coherent} and robust to noise \cite{chakrabarti_evolutionary_2006}. 
This problem of \textit{evolving community detection} has recently found myriad applications. 
For example, in social and computer networks, the continual estimation of community structure allows for robust containment of infection \cite{nguyen_dynamic_2014}.  
In brain networks, communities are believed to represent node collections dedicated to specialized functions such as memory or vision \cite{bullmore_economy_2012}; tracking changes (induced e.g. via hormonal fluctuations) in these collections illuminates how physiological factors influence cognition \cite{mueller_dynamic_2021}.

Many classical approaches to evolving community detection behave locally, matching network topology and/or partitions across adjacent snapshots \cite{cazabet_dynamic_2017}. 
More recently, a number of \textit{global} \cite{gauvin_detecting_2014, christopoulos_state_2022}, or \textit{cross-time} \cite{rossetti_community_2018}, approaches have arisen in the offline setting. 
These methods utilize information across all snapshots to obtain more stable communities of long-term temporal coherence \cite{rossetti_community_2018,cazabet_dynamic_2017}. 
Numerous techniques for static community detection have been extended to this global setting \cite{gauvin_detecting_2014, sun_community_2020, ghasemian_detectability_2016}. 
However, the well-studied family of \textit{spectral methods}, approaches based on clustering Euclidean node embeddings derived from the extreme end of a real matrix's spectrum \cite{alpert_spectral_1995, newman_finding_2006, zhang_multiway_2015, white_spectral_2005,  shi_normalized_2000, ng_spectral_2001, saade_spectral_2014}, is a collection of notable exceptions. \new{In addition to being among the most competitive \cite{yang2016comparative}, scalable \cite{zhang_detecting_2020}, and theoretically rich \cite{belkin_laplacian_2003, bolla_spectral_2015, fasino_algebraic_2014, chung_spectral_1997} approaches to static community detection, spectral methods are also among the most general, with the above definition including techniques for clustering signed \cite{mercado_spectral_2019, mercado_clustering_2016, kunegis_spectral_2010}, mixed-membership \cite{wahl_hierarchical_2015, zhang_detecting_2020}, directed \cite{satuluri_symmetrizations_2011, zhou_learning_2005}, multiview \cite{mercado_spectral_2019, dong_clustering_2013}, cocommunity-structured \cite{rohe_co_2016}, hierarchical \cite{laenen_nearly_2023}, hyper~\cite{zhou_learning_2006}, motif-based~\cite{underwood_motif_2020}, and higher-order \cite{grande2023topological, krishnagopal_spectral_2021} networks. It is thus of great interest to develop temporal methods for spectral community detection inheriting these properties.}

\paragraph{Contributions} We study a natural and versatile framework for temporally generalizing spectral methods for static community detection in a global fashion. 
Our approach is based on the observation (Section~\ref{sec:proposed-framework}) that any spectral method may be characterized in terms of a least squares-optimal \textit{low-rank approximation of some matrix} $\boldsymbol{M}$. 
In this light, it becomes natural to formulate the evolving spectral community detection problem using dynamics on the Grassmann manifold $\text{Gr}(d,k)$ of $k$-dimensional linear subspaces of $\mathbb{R}^d$ \cite{boumal_introduction_2023}. 
Identifying $\{G_i\}_{i=1}^T$ with a discrete trajectory on $\text{Gr}(d,k)$, the aforementioned notions of `stability' and `temporal coherence' may be interpreted as constraining this subspace trajectory to lie on some highly regular Grassmann curve. 
Recent subspace estimation literature \cite{blocker_dynamic_2023, hong_parametric_2016} suggests geodesic curves to be an elegant and effective candidate. 
Our approach is to thus learn the Grassmann geodesic that best models the data, capturing a robust continuum of node embeddings from which a sequence of clusterings may be computed. \new{Advantages include}:

\begin{itemize}
    \item Any community detection method based on the leading or trailing eigenvectors of a matrix admits temporal generalization, as shown in Section~\ref{sec:instantiations}. 
    This is achieved using  subspace estimation methods that leverage the regularity of Grassmann geodesics. 
    The proposed framework provides algorithms for community detection in a vast array of evolving networks, including evolving networks with weighted, signed, and/or directed edges, evolving networks with overlapping, hierarchical, or cocommunity structure, and evolving networks with multiple views (Table~\ref{tab:network-modalities}).
    \item \new{If the desired number of communities $k_c$ is known in advance, the framework produces algorithms capable of utilizing this information (Algorithm~\ref{alg:geodesic-dcd}). If not, a simple yet effective extension can ascertain $k_c$ at each time step automatically (Figure~\ref{fig:real}, Algorithm~\ref{alg:geodesic-dcd-variable-k})}.
    \item An intuitive heuristic (proven in Section \ref{sec:experiments}) allows practitioners to visualize the extent to which their data satisfies the framework's core assumption regarding Grassmann geodesic structure when $k_c=2$. Empirically, it indicates that smooth temporal evolution guarantees the presence of approximately geodesic structure in a dynamic stochastic block model (dynamic SBM), affirming that the geodesic assumption is a natural one.

    \item 
    The proposed method is remarkable for the performance gains it achieves across the diverse array of networks mentioned above.
    \new{For all network types considered, the adjusted mutual information and/or element-centric similarity} between the estimated and true dynamic communities is above 0.8 — and often very near 1.0 (perfect recovery) — on appropriate variants of the dynamic SBM and on real data. This represents significant gains over the static counterparts and other dynamic methods. The approach is robust to noise, reliably recovering planted community structure in \new{noise regimes} for which the corresponding static methods perform around random chance.
\end{itemize}

\begin{table}
\begin{tabular*}{5.5in}{@{\extracolsep{\fill}} c p{0.75\linewidth}}
\hline
\textbf{Notation} & \textbf{Description} \\
\hline
$G$ & A (simple, directed, multiview...) graph with $d$ nodes, suitably connected  \\
$\bs{A}$ & The (possibly weighted) adjacency matrix of $G$  \\
$\operatorname{deg}(\ell)$ & The degree of a node $\ell$ of $G$ defined as $\operatorname{deg}(\ell) = \sum_{j \neq \ell} A_{\ell j}$ \\
$\operatorname{vol}(S)$ & The volume of a node subset $S$,  defined as $\operatorname{vol}(S) = \sum_{v \in S} \operatorname{deg}(v)$ \\
$\bs{D}$ & The degree matrix $\operatorname{diag}(\operatorname{deg}(1), \dots, \operatorname{deg}(d))$ of $G$.\\
$\lambda(\bs{H}); \lambda_i(\bs H)$ & The spectrum of a matrix $\bs{H}$; the $i$th largest eigenvalue of $\bs H$. \\
$\sigma(\bs{H}); \sigma_i(\bs H)$ & The singular values of a matrix $\bs{H}$; the $i$th largest singular value of $\bs H$.\\

$\langle \bs{H} \rangle$ & The span of a matrix or vector $\bs{H}$ \\
$\text{St}(d,k)$ & The Stiefel manifold of matrices in $\mathbb{R}^{d \times k}$ with orthonormal columns \\
$\mathrm{O}_k$ & The orthogonal group of matrices in $\mathbb{R}^{k \times k}$ with orthonormal columns
\\
$\text{Gr}(d,k)$ & The Grassmann manifold of $k$-dimensional linear subspaces of $\mathbb{R}^d$ \\
$\mathscr{A}$ & An algorithm which clusters $d$ nodes with Euclidean embeddings into $\mathbb{R}^k$ derived from the extreme of a matrix spectrum (a \textit{static spectral method})
\\
$\bs M_k = \bs U_k \bs \Sigma_k \bs V_k^\top$ & The rank-$k$ singular value decomposition $\bs U_k(\bs M) \bs \Sigma_k(\bs M) \bs V_k(\bs M)^\top$ of $\bs M$ \\
MCM of $\mathscr{A}$ & A matrix $\bs M$ satisfying $\langle \bs U_k(\bs M) \rangle = \langle \bs C \rangle$, where $\bs C$'s rows are the spectral embeddings from $\mathscr{A}$ into $\mathbb{R}^k$. Called a \textit{modeled clustering matrix (MCM)}\\
\hline
\end{tabular*}
\end{table}


\paragraph{Related Work} Much early work on evolving community detection focuses on the online setting via the following `two-stage approach \cite{cazabet_dynamic_2017}': 

\vspace{-2mm}
\begin{enumerate}
    \item Detect communities per-snapshot using a static method, such as \cite{newman_finding_2004, blondel_fast_2008, rosvall_maps_2008, karrer_stochastic_2011} ;
\item Match communities across adjacent snapshots; we will call this \textit{local temporal smoothing} \cite{rossetti_community_2018}.  \label{local-temporal-smoothing} 
\end{enumerate}
For example, \cite{aynaud_static_2010} applies the Louvain method \cite{blondel_fast_2008} at each snapshot, imposing local temporal smoothing by initializing the method at snapshot $t$ with results from snapshot $t-1$.  The technique in \cite{morini_revealing_2017} computes similarity scores between the static communities detected at time $t$ and those at $t \pm 1,2$ and uses a `network sliding window' to smoothen away noise. In terms of spectral methods, a local temporal smoothing approach to spectral clustering is offered in \cite{chi_evolutionary_2007}. The authors apply spectral clustering per-snapshot, but with an added `temporal cost' parameter to the normalized cut objective function that quantifies how well a partition at time $t$ clusters the data at time $t-1$. They then relax their objective in a manner similar to that deriving static spectral clustering \cite{von_tutorial_2007}, and obtain an optimal solution to this relaxation at each time step in terms of eigenvectors of a new, `temporally smoothed' matrix. In what may be viewed as a substantial offline extension of this, \cite{liu_global_2018} offers an approach to temporal spectral clustering via the simultaneous estimation of $T$ smoothened Laplacian eigenbases $\bar{\boldsymbol{U}}_i$, $i \in [T]$, which are encouraged to be close both to their static counterpart $\boldsymbol{U}_i$ and smoothened predecessor $\bar{\boldsymbol{U}}_{i-1}$. This approach is `global' in the sense that its objective function aggregates community information across all time points, and hence estimates all communities simultaneously. At the same time, it is `local' in the sense that the objective only considers relationships between data in directly adjacent snapshots. In offline contexts where the full graph sequence is known \textit{a priori}, it is natural to pursue methods capable of capturing long-term temporal correlations \cite{christopoulos_state_2022}. Recent such global approaches include \cite{gauvin_detecting_2014}, which concatenates adjacency matrix snapshots into a rank-3 tensor and applies nonnegative tensor factorization techniques, generalizing \cite{yang_overlapping_2013}; \cite{sun_dynamic_2022}, which offers a dynamic model based on preferential attachment phenomena, generalizing \cite{sun_community_2020}; and 
\cite{ghasemian_detectability_2016}, which applies methods of statistical inference to the dynamic stochastic block model \cite{matias_statistical_2017}, generalizing techniques surveyed in \cite{abbe_community_2017}. Analogously, our work simultaneously generalizes to the global temporal setting methods such as spectral modularity maximization \cite{newman_finding_2006, white_spectral_2005, zhang_multiway_2015}, (un)normalized spectral clustering/graph partitioning  \cite{shi_normalized_2000, ng_spectral_2001, pothen_partitioning_1990}, Bethe Hessian clustering \cite{saade_spectral_2014}, and the multitudinous extensions of these to different network modalities. Different spectral approaches find utility in different contexts, and we envision accordingly diverse applications for our framework. All proofs are deferred to the appendix.

\vspace{-4mm}

\section{Proposed Framework}
\label{sec:proposed-framework}
Our proposed framework leverages the general assumption of spectral methods that the extremal eigenvectors of a given matrix provide a good embedding for clustering graph nodes \cite{alpert_spectral_1995, newman_finding_2006, zhang_multiway_2015, white_spectral_2005,  shi_normalized_2000, ng_spectral_2001, saade_spectral_2014}. 
We will show that these eigenvectors also span a subspace of best low-rank approximation to a modification of said matrix, and low-dimensional subspaces have a natural dynamic extension with curves on the Grassmann manifold of \new{$k$-dimensional subspaces \newnew{of $\mathbb{R}^d$}, which forms a compact \newnew{Riemannian manifold with metric inherited from} Euclidean space~\cite{boumal_introduction_2023}} . 
Our framework produces an extension of any static spectral method to the time-varying setting by applying a recent novel algorithm for fitting Grassmann geodesics \cite{blocker_dynamic_2023}. In this section, we first discuss the general template for spectral community detection in a single network instance.
A key feature is the connection of a subspace for node embeddings to a low-rank matrix approximation.
We then give our proposed method, which uses ideas from subspace tracking to guarantee those subspace node embeddings are smooth in time. 


Let $\{{G}_i\}_{i=1}^T$ be a multiplex graph whose layers represent snapshots of a time-varying network with $d$ nodes. The goal is to partition each $\boldsymbol{G}_i$ into $k$ communities — divisions into groups of increased connectivity \cite{newman_spectral_2013} — in a temporally coherent manner. 
For the case $T=1$, a number of successful \textit{static} community detection methods exist. 
Especially well-studied are the \textit{spectral methods}, which derive Euclidean node embeddings from the eigenvectors of a \textit{clustering matrix} $\bs R$, then employ a Euclidean clustering approach to obtain community assignments. Our proposed method relies on the fact that a majority of spectral methods may be (re)formulated per the following template, a claim we substantiate in Section~\ref{sec:instantiations} following the outline of our method in terms of said template in Section~\ref{sec:proposed-method}.
\begin{algorithm}[H]
\caption{Template for Spectral Community Detection in Static Networks}
\label{alg:static-scd-template}
\begin{algorithmic}[1]
\Require Graph $G$, number $k_c$ of communities to detect, embedding dimension $k_e$
\Require Spectral algorithm $\mathscr{A}$ for community detection, with clustering matrix $\boldsymbol{R}$ 
\State \textbf{Embedding:} Embed the nodes of $G$ into $\mathbb{R}^d$ as columns of some matrix $\bs M=\bs M(\boldsymbol{R}) \in \mathbb{R}^{d \times d}$ 

\State \textbf{Spatial denoising:} 
\label{step:low-rank-step}
Compress the node embeddings into a linear subspace $\mathcal{U} = \langle \bs U \rangle \in {\text{Gr}(d,k_e)}$,  $\bs U \in \text{St}(d,k_e)$, such that the reconstruction error $\|\bs M - \bs U \bs U^\top \bs M\|_{\mathrm{F}}^2$ is minimized

\State \textbf{Euclidean clustering:} 
\label{step:Euclidean-clustering}
Follow $\mathscr{A}$'s specifications to derive community assignments from $\boldsymbol{U}$

\Ensure An assignment $(Z_1,\dots,Z_k)$ of $G$'s nodes into $k_c$ communities. 

\end{algorithmic}
\end{algorithm}
The spectral nature of $\mathscr{A}$ is encoded in the subspace $\mathcal{U} \in \text{Gr}(d,k)$, as this subspace is represented by the eigenbasis $\bs U_k \in {\text{St}(d,k)}$ for $\bs M_k=\bs U_k \bs \Sigma_k \bs V_k^\top$ a rank-$k$ approximation to $\bs M$. We note, though, that any choice of orthonormal basis $\bs U$ of $\mathcal{U}$ suffices.\footnote{Technically, it is assumed here that the Euclidean clustering method of $\mathscr{A}$ (step~\ref{step:Euclidean-clustering}, e.g. $k$-means) is invariant under linear isometry of the input space. This is true for all spectral methods of which the authors are aware.} 
We call the matrix $\boldsymbol{M}$ a \textit{modeled clustering matrix (MCM) of $\mathscr{A}$}, for its defining property is that the rank-$k$ subspace $\langle \boldsymbol U_{k} \rangle$ optimally modeling its columns is provided by spectral embeddings from $\mathscr{A}$.

Constructing an MCM for a given spectral method $\mathscr{A}$ is not difficult; indeed,  Section~\ref{sec:instantiations} will show that choosing $\boldsymbol{M}=\boldsymbol{I} \pm \boldsymbol{R}/\|\boldsymbol{R}\|_{\mathrm{F}}$ suffices for all spectral methods of which the authors are aware, although other choices may yield more elegant interpretations. Finally, we remark that $k_c=k_e$ for the vast majority of algorithms. In principle they can differ, though. For example, when $k_c=2$ it is common to choose $k_e=1$, usually because $\mathscr{A}$ admits some natural derivation for this special case which exploits the \newnew{algebraic structure} of $\mathbb{R}$ in step~\ref{step:Euclidean-clustering} \cite{newman_finding_2006, von_tutorial_2007}. It has also been argued that choosing $k_e \geq k_c$ is superior for some $\mathscr{A}$ \cite{zhang_multiway_2015, alpert_spectral_1995, rebagliati2011spectral}.

\vspace{-2mm}








\subsection{Proposed Method} \label{sec:proposed-method} Now assume we are given a time-varying graph $\{G_i\}_{i=1}^T$, with the goal of recovering the true community structure at each time step in a manner that is robust to both spatial (i.e., in the high-dimensional embedding space $\mathbb{R}^d$) and temporal noise. 
The linear subspace modeling of step~\ref{step:low-rank-step} in Algorithm~\ref{alg:static-scd-template} aims to handle the former, but a regularity constraint must be applied to the sequence $\{\langle \bs U_i \rangle\}_{i=1}^T$ of subspaces in order to attain the latter. 
A natural choice is to seek a geodesic of best fit for the data on the Grassmann manifold, either by fitting a geodesic directly through the points $\{\langle \bs U_i \rangle\}_{i=1}^T$ or by learning a Grassmann curve whose $i$th sample  approximates $\bs M_i$ while obeying a geodesic-constrained trajectory. The former may be viewed as spatial denoising followed by temporal denoising; the latter as simultaneous spatiotemporal denoising. Our method can accommodate either interpretation, but we will limit focus to the latter.


Geodesics on the Grassmann manifold, like lines in Euclidean space, behave as parsimonious interpolating curves. A geodesic  $\bs U:[0,1] \to \text{Gr}(d,k)$ with starting point $\langle \bs H \rangle$, $\bs H \in \text{St}(d,k)$, in the direction of $\bs Y \in T_{\langle \bs H \rangle} \text{Gr}(d,k)$, may be parameterized as \cite{bendokat_grassmann_2024} \begin{equation}
    \bs U(t; \bs H, \bs Y, \bs \Theta) = \overbrace{\begin{bmatrix}\bs H & \bs Y\end{bmatrix}}^{=: \bs P}\overbrace{\begin{bmatrix}
        \cos(\bs \Theta t) \\ \sin(\bs \Theta t) 
    \end{bmatrix}}^{=: \bs C(t)}= \bs P \bs C(t),
\end{equation}
where $\bs \Theta \in \text{diag}(\mathbb{R} ^{k \times k})$ consists of principal angles between geodesic endpoints.\footnote{That is, $\boldsymbol{\Theta}=\cos^{-1}(\boldsymbol{S})$, where $\boldsymbol{S}$ is obtained via the singular value decomposition $\boldsymbol{Z} \boldsymbol{S} \boldsymbol{Q}^\top=\boldsymbol{U}(0)^{\top} \boldsymbol{U}(1)$.}  
Our goal is to learn geodesic parameters such that the aggregated reconstruction error \begin{align}
    \mathcal{L}\big(\bs U(t; \bs H,\bs Y,\bs \Theta) \big)= \sum_{i=1}^{T}\|\bs M_{i}-\bs U(t_{i})\bs U(t_{i})^{\top}\bs M_{i}\|_{\text{F}}^{2}\label{geodesic-objective}
\end{align} is minimized. Here, $t_i$ is the continuous time point assigned by the user to snapshot index $i$. For simplicity, we always assume $t_1, \dots, t_T$ are equally spaced along $[0,1]$. Unlike Euclidean lines, no closed-form solution for geodesic regression on the Grassmann manifold is known, and moreover the objective (\ref{geodesic-objective}) is nonconvex.
Following \cite{blocker_dynamic_2023}, we attempt to minimize ($\ref{geodesic-objective}$) via block coordinate descent alternating between optimizing $\bs P$ and optimizing $\bs \Theta$.
We include the algorithm steps here for completeness, and their derivations may be found in \cite{blocker_dynamic_2023}. 
We adopt this approach because it is hyperparameter-free and each update monotonically descends (\ref{geodesic-objective}) while converging to a global optimum in a majority of \cite{blocker_dynamic_2023}'s experiments. Additionally, the work in \cite{li2023convergence} \new{analyzed this algorithm and proved convergence to a stationary point.} 


\paragraph{$\bs P$ Update.} Assume $\bs \Theta$ is fixed along with prior iterate $\bs P^{(n)}$. Then (\ref{geodesic-objective}) is minimized by setting \begin{equation}
\label{eqn:P-update}
    \bs P^{(n+1)} = \bs W \bs V^\top,
\end{equation}
where $\bs W$ and $\bs V$ are obtained from the singular value decomposition $\bs W \bs \Sigma \bs V^\top$ of the $d \times 2k$ matrix

\vspace{-4mm}
\begin{equation}
    \sum_{i=1}^T \begin{bmatrix} \bs M_i \bs M_i^\top \bs P^{(n)} \bs D_i \cos (\bs \Theta t_i) & \bs M_i \bs M_i^\top \bs P^{(n)} \bs D_i \sin (\bs \Theta t_i)\end{bmatrix}
\end{equation}
\vspace{-2mm}

\paragraph{$\bs \Theta$ Update.} Suppose $\bs P$ is fixed along with prior iterate $\bs \Theta^{(n)}$. In contrast to the above, when $\boldsymbol \Theta$ is fixed and $\boldsymbol{P}$ is to be updated, no analytic expression for a minimizer of (\ref{geodesic-objective}) is known in this case. However, the objective is separable in the entries $\theta_1, \dots, \theta_k$ of $\boldsymbol{\Theta}$, and they can each be iteratively updated over $M$ iterations as \begin{equation}
    \theta_j^{(0)}:=[\boldsymbol{\Theta}^{(n)}]_{jj}; \ \theta_{j}^{(m+1)}=\theta_{j}^{(m)}-s^{(m)}{\sum_{i=1}^{T}\dot{f}_{i,j}\big(\theta_{j}^{(m)}\big)}, \text{ where }s^{(m)}=1 / {\sum_{i=1}^{T}w_{f_{i,j}}\big(\theta_{j}^{(m)}\big)}
\end{equation} may be interpreted as a (variable) gradient descent step size and \begin{align}
\dot{f}_{i,j}(\theta_{j})= &  \frac{t_{i}\sqrt{  {  (\alpha_{i,j}-\gamma_{i,j})^{2}} + 4 \beta_{i,j}^{2} }}{2}  \sin \left(2 \theta_{j} t_{i} - \phi_{i,j}\right) \quad ;  \\
w_{f_{i,j}}(\theta_{j}) = & \frac{\dot{f}_{i,j}(\theta_{j})}{-\frac{\pi}{2t_{i}}+ \left( \theta_{j} - \frac{\phi_{i,j}+\pi}{2t_{i}}   \right) \operatorname{mod} \frac{2\pi}{2t_{i}}}
\end{align}
for \begin{align}
\phi_{i,j}&=\operatorname{arctan2}\left(\beta_{i,j}, \frac{\alpha_{i,j} - \gamma_{i,j}}{2}\right) & \beta_{i,j}&=[\boldsymbol  Y^{\top} \boldsymbol  M_{i} \boldsymbol  M_{i}\boldsymbol  H]_{jj}\\
\alpha_{i,j} &= [\boldsymbol  H^{\top} \boldsymbol  M_{i} \boldsymbol  M_{i}^{\top} \boldsymbol  H]_{jj} & \gamma_{i,j}&= [\boldsymbol  Y^{\top} \boldsymbol  M_{i} \boldsymbol  M_{i} \boldsymbol  Y]_{jj}.
\end{align}
Here, $\operatorname{arctan2}(y,x)$ denotes the angle swept by the $x$-axis in $\mathbb{R}^2$ as it rotates counter-clockwise into the point $(x,y)$. We then set \begin{equation}
\label{eqn:Theta-update}
\boldsymbol{\Theta}^{(n+1)}=\operatorname{diag}(\theta_1^{(M)}, \dots, \theta_k^{(M)}).
\end{equation}

\paragraph{Initialization} Choosing $\bs P^{(0)}$ and $\bs \Theta^{(0)}$ amounts to choosing an initial geodesic to iteratively optimize. This in turn amounts to choosing geodesic endpoints based on the data $\{\bs M_i\}_{i=1}^T$. For our problem setting, we choose the rank-$k$ least squares-optimal subspaces approximating $\bs M_1$ and $\bs M_T$ — obtained via $k$-truncated singular value decompositions $\boldsymbol{M}_1=\bs H_1 \bs \Sigma_1 \bs K_1^\top$ and $\boldsymbol{M}_T=\bs H_T \bs \Sigma_T \bs K_T^\top$ — as said geodesic's initial point and end point respectively. This determines the direction $\bs Y$ to be $\boldsymbol{F}\boldsymbol{G}^\top$, where, writing $\boldsymbol Z \boldsymbol S \boldsymbol Q^{\top}=\boldsymbol H_1^{\top}\boldsymbol H_T$, $\boldsymbol F$ and $\boldsymbol G^{\top}$ come from a singular value decomposition $(\boldsymbol I- \boldsymbol H_1 \boldsymbol H_1^{\top} )\boldsymbol H_{T} \boldsymbol Q=\boldsymbol F \boldsymbol D \boldsymbol G^{\top}$. Thus, we initialize with \begin{equation}
\label{eqn:opt-init}
    \boldsymbol{P}^{(0)} := \begin{bmatrix}
        \boldsymbol{H}_1 & \boldsymbol{F} \boldsymbol{G}^{\top}
    \end{bmatrix} \quad \quad \boldsymbol{\Theta}^{(0)}:=\cos^{-1}(\boldsymbol{H}_1^\top \boldsymbol{H}_T).
\end{equation}


\begin{algorithm}[H]
\renewcommand{\algorithmicrequire}{\textbf{Input:}}
\renewcommand{\algorithmicensure}{\textbf{Output:}}
\caption{Evolving Community Detection with Grassmann Geodesics}
\label{alg:geodesic-dcd}
\begin{algorithmic}[1]
\Require Graphs and associated snapshot times $\{ G_i, t_i \}_{i=1}^{T}$, number $k_c$ of communities to track, dynamic embedding dimension $k_e$
\Require Static spectral method $\mathscr{A}$ admitting an MCM (i.e., $\mathscr{A}$ formulated as in algorithm~\ref{alg:static-scd-template})
\State \textbf{Embedding:} Follow $\mathscr{A}$ to form high-dimensional node embeddings as MCMs $\{ \boldsymbol M_i \}_{i=1}^{T}$ 
\State \textbf{Spatiotemporal denoising} Estimate a $\text{Gr}(d, k_e)$-geodesic ${\bs{U}}(t; \bs H, \bs Y,\bs \Theta)$ via the minimization of (\ref{geodesic-objective}) through alternating iterates (\ref{eqn:P-update}) and (\ref{eqn:Theta-update}), initialized per (\ref{eqn:opt-init}) 
\State \textbf{Euclidean clustering:} For each snapshot index $i$, follow $\mathscr{A}$  to obtain community assignments from low-dimensional Euclidean node embeddings derived from $\bs U(t_i)$
\Ensure Assignments $(Z_1 \dots, Z_{k_c})_i$ of the nodes in each of $\{G_i\}_{i=1}^T$ to communities
\end{algorithmic}
\end{algorithm}

\begin{table}
\caption{Example applications of static community detection across different network modalities, and associated spectral methods for pursuing it. SC abbreviates \textit{spectral clustering}.}
\label{tab:network-modalities}
\centering
\renewcommand{\arraystretch}{1.2}
\begin{tabular}{|p{0.13\textwidth}|p{0.26\textwidth}|p{0.52\textwidth}|}
\hline
\textbf{Modality} & \textbf{Example Application(s)} & \textbf{Methods} \\
\hline
\multirow{5}{*}[2.4ex]{\centering Simple} & Image segmentation \hfill\mbox{\footnotesize\cite{shi_normalized_2000}} & (Un)Normalized SC \hfill\mbox{(USC)(NSC) \footnotesize\cite{shi_normalized_2000, ng_spectral_2001}} \\
\cline{2-3}
& Brain connectivity \hfill\mbox{\footnotesize\cite{bullmore_complex_2009}} & Spectral Modularity Maximization \hfill\mbox{(SMM) \footnotesize\cite{white_spectral_2005, zhang_multiway_2015, newman_finding_2006}} \\
\cline{3-3}
& & Bethe Hessian Clustering \hfill\mbox{(BHC) \footnotesize\cite{saade_spectral_2014}} \\
\hline
\multirow{4}{*}[2.4ex]{\centering Signed} & Time series \hfill\mbox{\footnotesize\cite{aghabozorgi_time_2015}} & Signed Ratio SC \hfill\mbox{(SRSC) \footnotesize\cite{kunegis_spectral_2010}} \\
\cline{2-3}
& Voting networks \hfill\mbox{\footnotesize\cite{west_exploiting_2014}} & Geometric Mean SC \hfill\mbox{(GMSC) \footnotesize\cite{mercado_clustering_2016}} \\
\cline{3-3}
& & Signed Power Mean SC \hfill\mbox{(SPMSC) \footnotesize\cite{mercado_spectral_2019}} \\
\hline
\multirow{2}{*}{Overlapping} & Social networks \hfill\mbox{\footnotesize\cite{devi_analysis_2016}} & Overlapping SC \hfill\mbox{(OSC) \footnotesize\cite{zhang_detecting_2020}} \\
\cline{2-3}
& Neuronal networks \hfill\mbox{\footnotesize\cite{kim_detecting_2015}} & $c$-means SC \hfill\mbox{(CSC) \footnotesize\cite{wahl_hierarchical_2015}} \\
\hline
\multirow{4}{*}[2.4ex]{\centering Directed} & Genomics \hfill\mbox{\footnotesize\cite{popa_directed_2011}} & Degree-Discounted SC \hfill\mbox{(DDSC) \footnotesize\cite{satuluri_symmetrizations_2011}} \\
\cline{2-3}
& Information networks \hfill\mbox{\footnotesize\cite{huang_web_2006}} & Bibliographic SC \hfill\mbox{(BSC) \footnotesize\cite{satuluri_symmetrizations_2011}} \\
\cline{3-3}
& & Random Walk Directed SC \hfill\mbox{(RWSC) \footnotesize\cite{zhou_learning_2005}} \\
\hline
\multirow{2}{*}{Multiview} & Measuring networks \hfill\mbox{\footnotesize\cite{newman_network_2018}} & Grassmannian Multiview SC \hfill\mbox{(GMVSC) \footnotesize\cite{dong_clustering_2013}} \\
\cline{2-3}
& Multimedia analysis \hfill\mbox{\footnotesize\cite{chao_survey_2021}} & Power Mean Laplacian SC \hfill\mbox{(PMLSC) \footnotesize\cite{mercado_spectral_2019}} \\
\hline
Cocommunity & NLP \hfill\mbox{\footnotesize\cite{dhillon_co_2001}} & Spectral Coclustering \hfill\mbox{(SCC) \footnotesize\cite{rohe_co_2016}} \\
\hline
\multirow{2}{*}{Hierarchical} & Biochemical graphs \hfill\mbox{\footnotesize\cite{ravasz_hierarchical_2002}} & Hierarchical SC \hfill\mbox{(HSC) \footnotesize\cite{laenen_nearly_2023}} \\
\cline{2-2}
& Brain networks \hfill\mbox{\footnotesize\cite{ashourvan_multi_2019}} & \\
\hline
\end{tabular}
\end{table}

\subsection{Instantiations}
\label{sec:instantiations}
This section illustrates the generality of algorithm~\ref{alg:static-scd-template} by way of several examples, thereby providing a collection of static community detection methods extended by algorithm template~\ref{alg:geodesic-dcd} to the temporal setting. The (possibly weighted) simple network cases of spectral clustering and spectral modularity maximization are a faithful prototype for the exploration of other techniques using our framework. We therefore restrict present attention to these methods, with analogous discussion of other methods deferred to Appendix \ref{sec:instantation-detail}.
Table~\ref{tab:network-modalities} summarizes all static spectral methods that we analyze and evaluate. To begin, we observe that any spectral method based on the extremal eigenvectors of a clustering matrix $\boldsymbol R$ easily admits an MCM $\boldsymbol M$, simply via the normalization, then shifting, of its spectrum. 

\begin{proposition}
    \label{prop:mcm-prop}
    Let $\mathscr{A}$ be a static spectral algorithm with clustering matrix $\boldsymbol R$. If $\mathscr{A}$ clusters using the leading (resp. trailing) eigenvectors of $\boldsymbol R$, then $\bs M=\boldsymbol I + \boldsymbol R/\|\boldsymbol R\|_{\mathrm{F}}$ (resp. $\boldsymbol I - \boldsymbol R /\|\boldsymbol R\|_{\mathrm{F}}$) is an MCM for $\mathscr{A}$ (cf. Algorithm~\ref{alg:static-scd-template}).
\end{proposition}

Some alternative choices of $\boldsymbol M$ come with extra interpretations. See e.g. propositions \ref{prop:prop-low-rank-unnormalized-sc} and \ref{prop:signless-lapl} below. Proposition~\ref{prop:mcm-prop} suffices to generate MCMs for all spectral methods considered in this text, though, and will be employed by default when no alternative is specified. 


%



\paragraph{Unnormalized Spectral Clustering (USC)}
Let $\bs{A}$ be the (possibly weighted) adjacency matrix and $\bs{D}$ be the degree matrix of the graph $G$.
Unnormalized spectral clustering arises from the observation that the matrix $\bs V$ composed of the trailing $k$ eigenvectors of the graph Laplacian $\bs R = \boldsymbol{L}=\boldsymbol{D}-\boldsymbol{A}$ minimizes a relaxation of the NP-Hard $\operatorname{RatioCut}$ objective for partitioning a graph into communities of roughly balanced size \cite{simon_partitioning_1991}, such that the $k$-dimensional embedding of $G$ offered by $\bs V^\top$ unveils clustering structure in Euclidean space discoverable by heuristics like cardinality-constrained vector partitioning \cite{alpert_spectral_1995} or $k$-means \cite{von_tutorial_2007}. 

\begin{proposition}[]
\label{prop:prop-low-rank-unnormalized-sc}
  The matrix $\boldsymbol M=n\boldsymbol I - \boldsymbol L$ is an MCM for unnormalized spectral clustering for any $n \geq \max |\lambda(\boldsymbol L)|$. In particular, choosing any $n \geq 2\max_{i}\operatorname{deg}(i)$ — for example, $n=2d$ when $G$ is unweighted — yields an MCM. 
\end{proposition}

 A notable special case of Proposition~\ref{prop:prop-low-rank-unnormalized-sc} occurs when $G$ is a graph on $k$ disjoint cliques containing $s$ nodes each. In this case, the community structure of $G$ is entirely unambiguous and $\boldsymbol M=s\boldsymbol I - \boldsymbol L$ turns out to be precisely rank-$k$ (Appendix~\ref{sec:simple-networks}).

\paragraph{Normalized Spectral Clustering (NSC)}Normalized spectral clustering arises from a modification of the $\operatorname{RatioCut}$ objective wherein balance is encouraged among in-community edge densities rather than community sizes. The new objective's relaxation is minimized with the $k$ smallest eigenvectors of the \textit{normalized symmetric graph Laplacian} $\bs L^{\text{sym}}:=\bs D^{-1/2}\bs L\bs D^{-1/2}$. 

\begin{proposition}
\label{prop:signless-lapl}
The normalized signless Laplacian $\boldsymbol Q^{\text{sym}} = \boldsymbol D^{-1/2} \boldsymbol Q \boldsymbol D ^{-1 /2}$, $\boldsymbol Q=\boldsymbol D + \boldsymbol A$, is an MCM for normalized spectral clustering.
\end{proposition}

The tail of $\bs Q^{\text{sym}}$'s spectrum is intimately related to \textit{anti}community structure \cite{kirkland_bipartite_2011}. Proposition~\ref{prop:signless-lapl} complementarily relates the head of its spectrum to community structure. 

\paragraph{Spectral Modularity Maximization (SMM)} In many real-world networks, what matters is not that inter-community edge density is low, but that it is \textit{lower than expected}. The notion of network \textit{modularity} proposed in \cite{newman_finding_2004} quantifies this. Numerous spectral approaches to approximately maximize modularity exist by clustering embeddings derived from the leading positive eigenvectors of the \textit{modularity matrix} $\boldsymbol{B}=\boldsymbol{A}-\mathbb{E}\bs A$ \cite{newman_finding_2006, white_spectral_2005, zhang_multiway_2015}.\footnote{Here, the expectation is taken with respect to the \textit{Newman-Girvan Null Model} \cite{newman_finding_2006}: $[\mathbb{E}\bs A]_{j\ell} = {\operatorname{deg}(j)\operatorname{deg}(\ell)}/{\operatorname{vol}(A)}$. This is standard in contemporary treatments of network science such as \cite{newman_networks_2018}.} 

\vspace{-2mm}



\section{Experiments}


\begin{figure}
    \centering
    \begin{subfigure}[b]{0.49\textwidth}
        \centering
        \includegraphics[width=\textwidth]{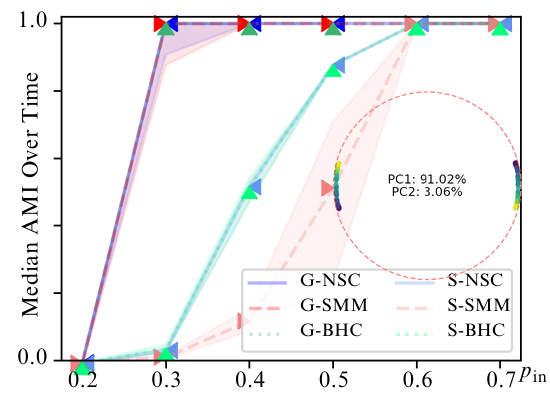}
        \caption{$p_{\text{switch}}=T/1000$}
        \label{subfig:pswitch-1000}
    \end{subfigure}
    \hfill
    \begin{subfigure}[b]{0.49\textwidth}
        \centering
        \includegraphics[width=\textwidth]{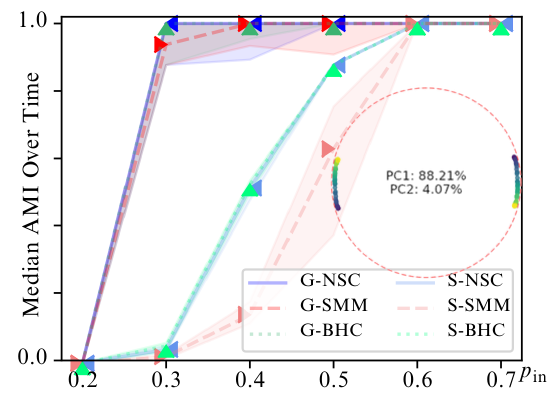}
        \caption{$p_{\text{switch}}=T/750$}
        \label{subfig:pswitch-750}
    \end{subfigure}
    
    \vspace{1em}
    
    \begin{subfigure}[t]{0.49\textwidth}
        \centering
        \includegraphics[width=\textwidth]{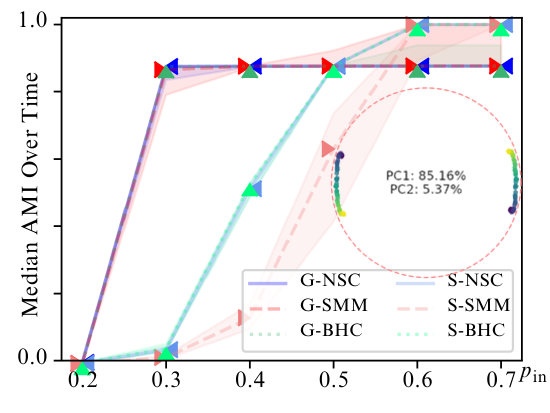}
        \caption{$p_{\text{switch}}=T/500$}
        \label{subfig:pswitch-500}
    \end{subfigure}
    \hfill
    \begin{minipage}[b]{0.49\textwidth}
        \centering
    \captionof{figure}{Comparison of median $\text{\new{AMI}} \in [0,1]$ over time versus $p_{\text{in}}$, medianed over $50$ simulations of the dynamic SBM ($d=T=50$, $k=2$, $p_{\text{out}}=0.2$) as $p_{\text{switch}}$ ranges from low to very high. Each unit circle displays the trajectory of modularity matrix first-eigenvectors as $T$ progresses when $p_{\text{in}}=0.4$. For low (\ref{subfig:pswitch-1000}) and medium (\ref{subfig:pswitch-750}) values of $p_{\text{switch}}$, said trajectory `walks along' the unit circle, suggesting (Proposition~\ref{prop:geodesic-assumption}) that the dynamics satisfy the geodesic assumption. When $p_{\text{switch}}$ is very high (\ref{subfig:pswitch-500}), the trajectory `falls off': the assumption has been violated. G- resp. S- refers to geodesic resp. static algorithm versions. }
        \label{fig:vary-pswitch}
    \end{minipage}
\end{figure}

\label{sec:experiments}
We demonstrate the effectiveness of our method on real and synthetic evolving network data. On the synthetic data, we relate the behavior of our method to the presence of latent geodesic structure as graph dynamics range from smooth to jagged. We also show that the geodesic generalizations of $15$ spectral methods (Table~\ref{tab:network-modalities}) via Algorithm~\ref{alg:geodesic-dcd} empirically outperform their static counterpart, where the static method is applied separately at each time point. On real data, we show that the geodesic method achieves favorable performance over a collection of popular temporal community detection benchmarks, including when the algorithm no longer assumes a fixed number of temporally stable latent communities. Experiments were all performed in \texttt{Python} on a 2019 MacBook Pro.\footnote{\newnew{Code is available at \href{https://github.com/jacobh140/spectral-dcd}{https://github.com/jacobh140/spectral-dcd}.}} 

\paragraph{Model selection and checking the geodesic assumption} Throughout this section, model selection (i.e., choosing $k_e$ and/or $k_c$) is performed by letting $k$ range over successive algorithm runs and choosing that which yields partitions of highest modularity \cite{white_spectral_2005}. We quantify `partitions of highest modularity' by taking the mode over $k$, but other summarizations such as multilayer modularity \cite{mucha_community_2010} could be used as well. 
%
 %
When $k_c=2$, an additional step is often available for assessing the presence of geodesic structure.
 For clarity of exposition, we focus on spectral modularity maximization, $\bs M=\overline{\bs B}=\boldsymbol{I} + \boldsymbol{B}/\|\boldsymbol{B}\|_{\mathrm{F}}$. 
 When the goal is two-way clustering, it is customary to classify node $\ell$ based on $\operatorname{sign}(u_\ell)$, where $\boldsymbol{u}$ is the first eigenvector of $\overline{\boldsymbol{B}}$ \cite{newman_finding_2006}. This corresponds to considering the space $\text{Gr}(1,d)$ in Algorithm~\ref{alg:geodesic-dcd} (i.e., $k_e=1$, $k_c=2$). Of course, the lines comprising $\text{Gr}(1,d)$ are naturally identified with pairs of antipodal points on the $(d-1)$-sphere $\mathbb{S}^{d-1} \subset \mathbb{R}^d$, and geodesics on $\mathbb{S}^{d-1}$ coincide with great circle arcs. The main idea, then, is that if we concatenate the respective first eigenvectors $\boldsymbol{u}_1, \dots, \boldsymbol{u}_T$ of $\overline{\boldsymbol{B}}_1, \dots, \overline{\boldsymbol{B}}_T$ and their negations $-\boldsymbol{u}_1, \dots, -\boldsymbol{u}_T$ into a data matrix $\boldsymbol{X} \in \mathbb{R}^{d \times 2T}$ and project these columns into their 2-dimensional PCA subspace, we should see the projections of the $\bs u_i$ into $\mathbb{R}^2$ `walk along' the unit circle in two mirrored trajectories if geodesic structure is present. The following proposition makes this precise. Results can be seen in Figure~\ref{fig:vary-pswitch}. \begin{proposition}
    $\sigma_{1}(\bs X) \geq \sigma_{2}(\bs X) > \sigma_{3}(\bs X)=\dots=\sigma_{\min(d, 2T)}(\bs X)=0$ if and only if the first singular subspaces of the $\overline{\boldsymbol B}_i$ lie in the image of a curve in $\text{Gr}(d,1)$ that is a Riemannian geodesic.
    \label{prop:geodesic-assumption}
\end{proposition}

\begin{figure}
    \centering
\includegraphics{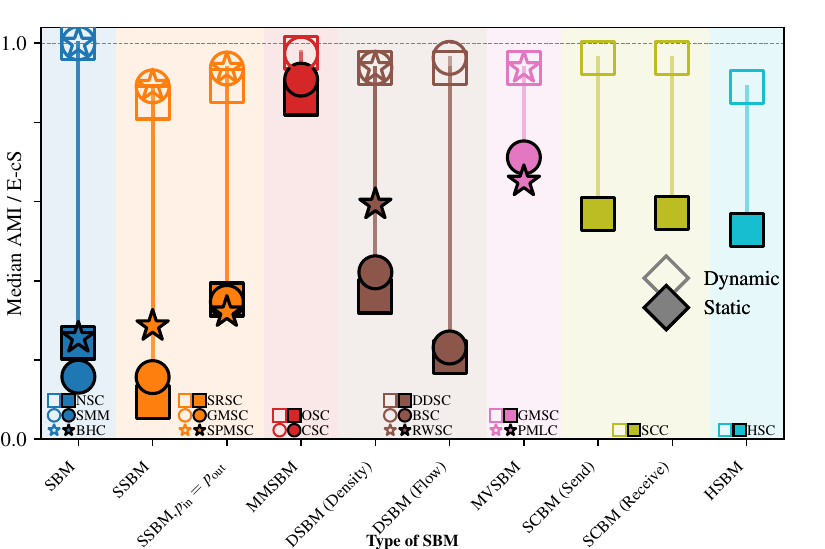}
    \caption{Comparison of median (over $50$ simulations and $20$ time steps of the appropriate SBM) \new{AMI/E-cS} for various static spectral methods and their dynamic generalization. Each color corresponds to a network modality, each line an SBM setting for that modality, and each symbol a spectral method for detecting communities in that modality (hollow for its dynamic extension, filled for static). By default where applicable, $d=120$, $k=2$, $p_{\text{in}}=0.3$, $p_{\text{out}}=0.2$, and $p_{\text{switch}}=10^{-2}$.
    Exceptions are in the second SSBM and second DSBM parameter settings, where $p_{\text{in}}=p_{\text{out}}$ to enforce clustering based solely on edge affinity/orientation. 
    Appendix~\ref{sec:instantation-detail} elaborates upon each individual column.}
    \label{fig:results-summary}
\end{figure}

\paragraph{Synthetic} We evaluate with a dynamic model based on a setting of the dynamic stochastic block model (dynamic SBM) studied in \cite{matias_statistical_2017}. Initially, $d$ nodes are divided equally among $k$ planted communities. At a given snapshot index $i \in [T]$, an edge is placed between each pair of intra-community nodes resp. inter-community nodes with probability $p_{\text{in}}$ resp. $p_{\text{out}}$. Dynamics are introduced by stipulating that at each snapshot, any node (that has yet to switch) switches community with probability $p_{\text{switch}}$. The performance of algorithms for simple networks, \new{quantified using adjusted mutual information (AMI) \cite{vinh2009information},} is compared in Figure~\ref{fig:vary-pswitch}, which also assesses geodesic structure (using Proposition~\ref{prop:geodesic-assumption}) as SBM dynamics range from smooth to very jagged.
Algorithms for non-simple network modalities are evaluated using a dynamic extension of the appropriate SBM. That includes the signed SBM (SSBM) for signed networks, mixed-membership SBM (MMSBM), directed SBM (DSBM), multiview SBM (MVSBM), stochastic coblock model (SCBM), and hierarchical SBM (HSBM). The definition of each modified SBM differs minutely from that of the standard SBM outlined here and is provided in appendix~\ref{sec:instantation-detail}.  Figure \ref{fig:results-summary} shows these results: in all cases, the communities recovered by the Grassmannian-smoothed dynamic network embeddings outperform those from the static embeddings, sometimes dramatically. \new{Overlapping and hierarchical community detection algorithms are evaluated using element-centric similarity (E-cS) \cite{gates2019element}.}


\paragraph{Real} We assess performance on the face-to-face interaction network of \cite{stehle_high_2011}. Over the course of two elementary school days, $77,602$ contact events were recorded among $242$ individuals ($232$ children in $10$ classes, and the $10$ teachers of those classes). The data was segmented into $10$-minute intervals, yielding a sequence of 
$T=102$ contact networks encoded by unweighted symmetric $232 \times 232$ adjacency matrices. As preprocessing, bridge edges were placed between graph components to maintain connectedness and teachers were removed from the network. The class memberships lend natural community structure to this evolving network, which are perfectly recovered at each time step by geodesic spectral clustering and near-perfectly recovered by the other geodesic methods (Figure~\ref{fig:real}, LEFT). 
We benchmark against the suite of popular dynamic community detection algorithms implemented within the \texttt{tnetwork} library \cite{cazabet_documentation_nodate}, \new{specifically the Label Smoothing (LS)  \cite{falkowski_mining_2006}, Smoothed Louvain (SL) \cite{aynaud_static_2010}, and Graph Smoothing (SG) \cite{guo_evolutionary_2014}} algorithms.  
The three geodesic approaches each outperform the benchmarks across all time, apparently dramatically so around the $3$-$\operatorname{4h}$ point — presumably a lunch break — on each day. We refrain from making a claim as to whether the mingling of classes around lunchtime is best considered as a noisy sample or as the dissolution of the latent community structure, and we indicate our ambivalence by masking these intervals in gray. We hypothesized that the geodesic method's favorable performance could be attributed to it seeking precisely $k_c=10$ communities at each time step while the benchmarks seek a variable number of communities. This turned out to be largely false: it is indeed the case that when stable latent communities are anticipated in the network dynamics, a dynamic community detection algorithm  capable of utilizing this information such as Algorithm~\ref{alg:geodesic-dcd} is preferred. However, (Figure~\ref{fig:real}, RIGHT) shows that a simple extension of Algorithm~\ref{alg:geodesic-dcd} to the case of varying $k_c$, based on sweeping $k_c$ to locally maximize modularity (Appendix~\ref{sec:variable-k-appendix}), performs nearly identically to Algorithm~\ref{alg:geodesic-dcd}. 
\vspace{-4mm}




\begin{figure}
    \centering
    \begin{subfigure}[b]{0.45\textwidth}
        \centering
        \includegraphics{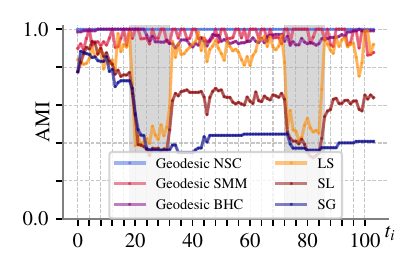}
        \label{subfig:eval}
    \end{subfigure}
    \hfill
    \begin{subfigure}[b]{0.45\textwidth}
        \centering

        \includegraphics{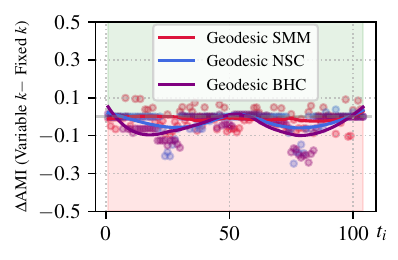}
        \label{subfig:variable-vs-fixed-k}
    \end{subfigure}
    \caption{Evaluation on the two-day elementary school face-to-face interaction network of \cite{stehle_high_2011}. LEFT: Included as benchmarks are the Label Smoothing (LS) approach of \cite{falkowski_mining_2006}, the Smoothed Louvain (SL) algorithm of \cite{aynaud_static_2010}, and the Graph Smoothing (SG) approach of \cite{guo_evolutionary_2014}. The three geodesic approaches uniformly outperform the benchmarks. RIGHT: The \new{AMI} difference at each time step between Algorithm~\ref{alg:geodesic-dcd} (fixed $k_c$) and its extension (Appendix~\ref{sec:variable-k-appendix}) to the variable-$k_c$ case. The data points are scattered, with each trajectory smoothened using a third-order Savitzky-Golay filter \cite{savitzky1964smoothing}  to visually enhance any trends. The results imply that the two algorithms perform very similarly. 
    }
    \label{fig:real}
\end{figure}

\section{Conclusion} 
\label{sec:conclusion}
This work presented and analyzed a Grassmann geometry-based framework for generalizing the popular family of spectral algorithms for community detection in static networks to the time-varying setting. 
Our method is broadly applicable to spectral community detection methods and has excellent performance, as demonstrated on both synthetic and real data experiments.
We note that our method only focuses on a single approach to fitting a geodesic to data. Extensions and improvements of said approach would induce analogous extensions and improvements here. 

\section{Acknowledgments}
The authors would like to thank Mark Newman, Samuel Sottile, Jeffrey Fessler, Soo Min Kwon, Haroon Raja, Cameron Blocker, and Alexander Saad-Falcon for their insightful discussions. J. Hume was supported in part by an REU associated with NSF CAREER award CCF-1845076. L. Balzano was supported in part by NSF CAREER award CCF-1845076 and DoE award DE-SC0022186.

\bibliographystyle{unsrtnat}
\bibliography{gs_bib_test}

\clearpage
\appendix

\renewcommand{\contentsname}{\newnew{Table of Contents}}
\newnew{\tableofcontents}

\clearpage

\section{Instantiations and Evaluations of the Proposed Framework}

\label{sec:instantation-detail}

This section provides proofs, interpretations, and experiments expounding the discussion in Section~\ref{sec:instantiations}. We begin with the straightforward proof of Proposition~\ref{prop:mcm-prop}. 

\begin{proof}[Proof of Proposition \ref{prop:mcm-prop}]
    First suppose $\mathscr{A}$ clusters based on the leading eigenvectors of $\boldsymbol R$. Since $|\lambda_d(\boldsymbol  R)| \leq \|\boldsymbol  R\|_{\text{F}}$, scaling by $1/\|\boldsymbol  R\|_{\text{F}}$ normalizes the spectrum of $\boldsymbol  R$ to lie in $[-1,1]$. Then the addition of $\boldsymbol  I$ turns the matrix positive semidefinite, whence its singular vectors and leading eigenvectors agree. 
    Now suppose $\mathscr{A}$ clusters based on the trailing eigenvectors of $\boldsymbol R$. Analogous to the above, division by $\|\boldsymbol R\|_{\mathrm{F}}$ normalizes the spectrum of $\boldsymbol R$ to lie in $[-1,1]$. Then subtracting it from $\boldsymbol I$ turns the matrix positive semidefinite, such that the top singular vectors of the matrix $\boldsymbol I - \boldsymbol R/\|\boldsymbol R\|_{\mathrm{F}}$ agree with its leading eigenvectors, which in turn agree with the trailing eigenvectors of $\boldsymbol R$. 
\end{proof}

We now proceed to discuss specific network modalities. The figures shown in this section are the numerical results used to create Figure~\ref{fig:results-summary} in the main text.

\subsection{Simple Networks}
\label{sec:simple-networks}
\begin{proof}[Proof of Proposition \ref{prop:prop-low-rank-unnormalized-sc}]

Since $\boldsymbol{L}$ is symmetric, its spectrum is real — say, $\lambda(\boldsymbol  L) \in [\lambda_{\min}, \lambda_{\max}] \subset \mathbb{R}$. $\max |\lambda(\boldsymbol L)|$ either equals $|\lambda_{\min}|$ or $|\lambda_{\max}|$; the spectrum of the matrix $\boldsymbol L / \max |\lambda(\boldsymbol L)|$ therefore lives in $[-1,1]$, and so we have  $\lambda(n \boldsymbol I - \boldsymbol L) \subset [0,2]$ for any $n \geq \max |\lambda(\boldsymbol L)|$. Since this matrix is positive semidefinite, its top singular vectors and leading eigenvectors agree. Of course, its leading eigenvectors are precisely the trailing eigenvectors of $\boldsymbol L$. 
    
For the second portion of the statement, let $n \geq 2 \max_{i \in V}\operatorname{deg}(i)$. Given a node $i \in V$, notice that $$\sum_{j \neq i} |L_{ij}|=\sum_{j \neq i} |-A_{ij}|=\operatorname{deg}(i).$$
By the Gershgorin circle theorem (\cite{fessler_linear_2024}, Section 8.5.2), the eigenvalues $\lambda_{1} \geq \lambda_{2} \geq \dots \geq \lambda_{d} \geq 0$ of $\boldsymbol L$ are located within the union of closed discs \[
    \lambda_i \in \bigcup_{i=1}^d \{\lambda \in \mathbb{R} : |\lambda - \operatorname{deg}(i)| \leq \operatorname{deg}(i)\} 
\]
This implies $\lambda(\boldsymbol L) \subset [0, 2 \max_{i \in V}\operatorname{deg}(i)]$ and thus $\lambda(\frac{\boldsymbol L}{n}) \subset [0,1]$. By the same reasoning that concludes the paragraph above, the result follows.
\end{proof}

\paragraph{Ideal communities and subspace structure} We now elaborate on the special case of Proposition~\ref{prop:prop-low-rank-unnormalized-sc} alluded to in Section~\ref{sec:instantiations}, wherein $G$ consists of $k$ cliques of $s$ nodes each. 

Permuting the node indices so that the graph Laplacian $\boldsymbol{L}$ of $G$ is block diagonal, with each block corresponding to a clique and identical to each other block, computing the spectrum of $\bs L$ amounts to computing the spectrum corresponding to a single complete graph $K_{s}$ of size $s$ and copying that spectrum $k$ times. The graph Laplacian $\boldsymbol L(K_{s})$ of $K_{s}$ is $$\boldsymbol  L_{K_{s}}=(s-1)\boldsymbol  I - (\boldsymbol  1 \boldsymbol  1^{\top}-\boldsymbol  I)=s\boldsymbol I - \boldsymbol 1 \boldsymbol 1^{\top},$$
where $\boldsymbol 1$ denotes the vector of all ones. As with any graph Laplacian, the vector $\boldsymbol 1$ is an eigenvector with eigenvalue $0$. Any other eigenvector $\boldsymbol v$ (of which there are $s-1$) must be orthogonal to this one, implying that $$\boldsymbol  1 \boldsymbol  1^{\top} \boldsymbol  v=\boldsymbol  1\left( \sum_{j}v_{j} \right)=\boldsymbol  1 \cdot 0=0.$$
It follows that $$\boldsymbol L_{K_{s}}\boldsymbol v=  s \boldsymbol  I\boldsymbol  v - \boldsymbol  1 \boldsymbol  1^{\top}\boldsymbol  v=s \boldsymbol  v,$$
i.e., $\boldsymbol v$ has eigenvalue $s$. Putting together each spectrum, we obtain $$\lambda(\boldsymbol  L)=(\overbrace{0,\dots,0}^{k \text{ times}}, \overbrace{s, \dots, s}^{d-k \text{ times}}).$$
Then the matrix $\boldsymbol M = s \boldsymbol I - \boldsymbol L$ has spectrum $$\lambda(\boldsymbol  M)=s - \lambda(\boldsymbol  L)=(\overbrace{s, \dots, s}^{k \text{ times}}, \overbrace{0,\dots,0}^{d-k \text{ times}}).$$
The result is that the MCM $\boldsymbol M$ has rank-$k$, meaning that in this special case the seemingly high-dimensional node embeddings it contains as columns already live in a low-dimensional latent subspace. This lends extra intuition to the formulation of spectral clustering as a subspace estimation problem (Algorithm~\ref{alg:static-scd-template}): in the case of a graph with entirely unambiguous community structure, the embeddings from $\boldsymbol M$ live in a $k$-dimensional subspace. We can then view small perturbations of this ideal graph structure — the random removal of edges between intracluster nodes and placement of edges between interclique nodes —  as adding spatial noise to the original node embeddings living in the subspace, which it is then the goal of spectral clustering to remove (Algorithm~\ref{alg:static-scd-template}) via low-rank approximation.

\begin{proof}[Proof of Proposition \ref{prop:signless-lapl}]
    For a discussion on signless Laplacians, see \cite{cvetkovic_signless_2007}. The only property needed here is that, like $\boldsymbol L= \boldsymbol D-\boldsymbol A$, $\boldsymbol Q=\boldsymbol D + \boldsymbol A$ is positive semidefinite.  It follows that the matrices
\begin{align}
\boldsymbol  {L}^{\text{sym}} = \boldsymbol  {D}^{-1/2} (\boldsymbol  {D} - \boldsymbol  {A}) \boldsymbol  {D}^{-1/2} \text{ and } \boldsymbol  Q^{\text{sym}} = \boldsymbol  {D}^{-1/2} (\boldsymbol  {D} + \boldsymbol  {A}) \boldsymbol  {D}^{-1/2} = 2\boldsymbol  {I} - \boldsymbol  {L}^{\text{sym}}
\end{align}
are each positive semidefinite. So, \begin{equation}
\lambda(\boldsymbol  {L}^{\text{sym}}) \subset \mathbb{R}_{\geq 0} \text{ and } 2 - \lambda(\boldsymbol  {L}^{\text{sym}}) \subset \mathbb{R}_{\geq 0},
\end{equation} which constrains $\lambda(\boldsymbol {L}^{\text{sym}})$ to the interval $[0,2]$. It follows that $\lambda(  \frac{1}{2} \boldsymbol L_{\text{sym}} ) \subset [0,1]$, and in turn that $\lambda(\boldsymbol I- \frac{1}{2} \boldsymbol L^{\text{sym}}) \subset [0,1]$. In particular, $\boldsymbol I - \frac{1}{2}\boldsymbol L^{\text{sym}}$ is positive semidefinite, hence its top singular vectors $\boldsymbol u_{1},\dots, \boldsymbol u_{k}$ agree with its leading eigenvectors $\boldsymbol v_{1},\dots, \boldsymbol v_{k}$. These are precisely the $k$ trailing eigenvectors of $\frac{1}{2}\boldsymbol L^{\text{sym}}$, and precisely the top singular vectors of $\frac{1}{2} \boldsymbol Q^{\text{sym}}$; the result follows.
\end{proof}

\paragraph{Bethe Hessian clustering} 
There is one additional static spectral method for clustering simple networks whose discussion is largely omitted from the main text. Bethe Hessian clustering \cite{saade_spectral_2014} generalizes spectral clustering by replacing the graph Laplacian with a regularized generalization $\boldsymbol{H}(r)=(r^2 - 1)\boldsymbol{I}-r \boldsymbol{A} + \boldsymbol{D}$. We set $r=\sqrt{c}$, where $c$ is the mean degree of $G$, per the suggestion in \cite{saade_spectral_2014}. Unlike the graph Laplacian, the Bethe Hessian is not in general positive semidefinite. Like the modularity matrix $\bs B$, the most positive and most negative eigenvectors in fact each contain important information about (anti)community structure. Here, we consider Bethe Hessian clustering (BHC) via the trailing eigenvectors of $\bs H(r)$, with MCM obtained (as always when no alternative is specified) via Proposition~\ref{prop:mcm-prop}.

\begin{SCfigure}
    \includegraphics{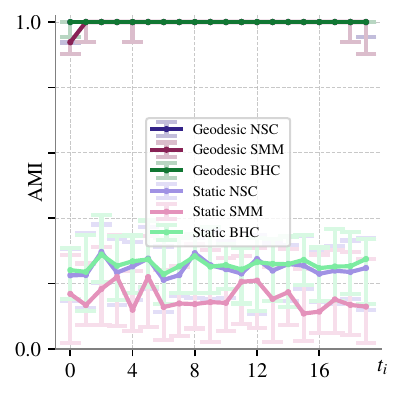}
    \caption{Representative comparison of geodesic and static methods for dynamic community detection in simple networks, medianed over $50$ simulations of the dynamic stochastic block model from Section~\ref{sec:experiments} ($d=120$, $T=20$, $p_{\text{switch}}=10^{-2}$, $p_{\text{in}}=0.3$, $p_{\text{out}}=0.2$). Error bars correspond to $25$th and $75$th percentiles. With the exception of spectral modularity maximization at time $i=1$, the geodesic methods each median to recover the true community structure of each snapshot. The methods compared are normalized spectral clustering (NSC) \cite{shi_normalized_2000}, spectral modularity maximization (SMM) \cite{white_spectral_2005}, and Bethe Hessian clustering (BHC) \cite{saade_spectral_2014}. } 
    \label{fig:sbm}
\end{SCfigure}

\subsection{Signed Networks}
\label{sec:signed-networks}
 Let $G=(V, E^+, E^-)$ be a graph comprised of both positive (attractive) and negative (repulsive) edges, such as voting networks \cite{west_exploiting_2014} or the Pearson correlation networks ubiquitous in time series analysis \cite{aghabozorgi_time_2015}. The study of community detection in {signed networks} has its roots in structural balance theory \cite{heider_attitudes_1946}; the goal is to obtain a partition under which increased positive edges exist within communities and increased negative edges exist between communities \cite{cucuringu_sponge_2019}.
 \paragraph{Signed networks — algorithms and modeled clustering matrices} Two generalizations of spectral clustering to signed networks are based on trailing eigenvectors of the \textit{signed ratio Laplacian} \begin{equation}
    |\bs D| - \bs A,
\end{equation}
where $| \bs D |$ is the entrywise absolute value of $\bs D$ \cite{kunegis_spectral_2010}, and the \textit{geometric mean Laplacian} \cite{mercado_clustering_2016}\footnote{Here, the $+$ and $-$ exponents refer to quantities defined over the unsigned networks $G^+=(V, E^+)$ and $G^- = (V, E^-)$ respectively. }  \begin{equation}
    \boldsymbol L^{+^{1/2}}(\boldsymbol L^{+^{-1/2}} \boldsymbol Q^{-} \boldsymbol L^{+^{-1/2}})^{1/2}\boldsymbol L^{+^{1/2}}
\end{equation}
(or their normalized counterparts). More generally, choose $p \in \mathbb{R}$. Set $k':=k-1$ if $p \geq 1$ and $k'=k$ if $p < 1$. The approach in \cite{mercado_spectral_2019} subsumes both of the aforementioned approaches as the special cases $p=1$ and $p \to 0$ of applying $k$-means to the $k'$ smallest eigenvectors of the \textit{signed power mean $p$-Laplacian} \begin{equation}
    \boldsymbol{L}(p) := \bs M_{p}(\boldsymbol L^{\text{sym}+}, \boldsymbol Q^{\text{sym}-}),
\end{equation} 
where $\boldsymbol{M}_p(\bs A, \bs B)$ denotes the matrix power mean $\boldsymbol M_{p}(\boldsymbol A, \boldsymbol B):=(\frac{\boldsymbol A^{p}+\boldsymbol B^{p}}{2})^{1/p}$.\footnote{When $p<0$ the matrix power mean requires positive definite matrices; the authors address this by considering $\bs L^{\text{sym}+} + \varepsilon \boldsymbol{I}$ and $\bs Q^{\text{sym}-} + \varepsilon \boldsymbol{I}$ for some $\varepsilon > 0$ as necessary. }

\begin{proposition}[]
\label{prop:signed-low-rank}
    Let $\bs U_{k'} \bs \Sigma_{k'} \bs V_{k'}^\top$ be a truncated singular value decomposition of $\overline{\bs L}(p) := \boldsymbol I-  \frac{1}{2} \bs M_{p}(\boldsymbol L^{\text{sym}+}, \boldsymbol Q^{\text{sym}-})$. Take the leading singular vectors $\bs u_1, \dots, \bs u_{k'}$ as node embeddings. Then the matrix serves as an MCM for the three algorithms described here.
\end{proposition}
\begin{proof}[Proof of Proposition \ref{prop:signed-low-rank}]
    $M_{p}(\boldsymbol L^{\text{sym}+}, \boldsymbol Q^{\text{sym}-})$ is positive semidefinite as a combination of sums, powers, and positive scalings thereof. Hence $\lambda(M_{p}(\boldsymbol L^{\text{sym}+}, \boldsymbol Q^{\text{sym}-})) \subset \mathbb{R}_{\geq 0}$. Also, we can bound its leading eigenvalue as:
\begin{align}
\lambda_{1}(M_{p}(\boldsymbol L^{\text{sym}+}, \boldsymbol Q^{\text{sym}-})) = & \lambda_{1}\left(\left(\frac{\boldsymbol  L^{\text{sym}+^{p}} + \boldsymbol  Q^{\text{sym}-^{p}}}{2}\right)^{1 / p}\right) \\
= & \frac{1}{2} \lambda_{1}^{1 / p}\left(  \boldsymbol  L^{\text{sym}+^{p}} + \boldsymbol  Q^{\text{sym}-^{p}}  \right)  \\
\leq &\frac{1}{2} \lambda_{1}^{1/p} (\boldsymbol  L^{\text{sym}+^{p}}) + \frac{1}{2}\lambda_{1}^{1/p}(\boldsymbol  Q^{\text{sym}-^{p}})  \label{eq:prop5weyl} \\
=  & \frac{1}{2} \lambda_{1}^{p / p}(\boldsymbol  L^{\text{sym}+}) + \frac{1}{2} \lambda_{1}^{p/p}(\boldsymbol  Q^{\text{sym}-})\\
\leq& \frac{1}{2} \left( 2 \right) + \frac{1}{2}(2) \label{eq:prop5eigbound} \\
= & 2,
\end{align}
where the inequality in \eqref{eq:prop5weyl} follows from Weyl's inequality and the inequality in \eqref{eq:prop5eigbound} follows from the fact that both $\lambda(\boldsymbol L^{\text{sym}+}) \subset [0,2]$ and $\lambda(\boldsymbol Q^{\text{sym}-}) \subset [0,2]$ (cf. Proposition~\ref{prop:signless-lapl}). Hence $\lambda(M_{p}(\boldsymbol L^{\text{sym}+}, \boldsymbol Q^{\text{sym}-})) \subset \mathbb{R}_{ \leq 2}$. Thus $\lambda(M_{p}(\boldsymbol L^{\text{sym}+}, \boldsymbol Q^{\text{sym}-})) \subset [0,2]$. Or, if $p<0$ forces us to spectrum-shift by $\varepsilon>0$, $[\varepsilon, 2+\varepsilon]$. In turn, $\lambda\left(  \frac{1}{2} M_{p}(\boldsymbol L^{\text{sym}+}, \boldsymbol Q^{\text{sym}-}) \right) \subset [0,1]$; it follows that the matrix $\boldsymbol I-  \frac{1}{2} M_{p}(\boldsymbol L^{\text{sym}+}, \boldsymbol Q^{\text{sym}-})$ is positive semidefinite and hence its top singular vectors and top eigenvectors agree. But its top eigenvectors are precisely the bottom eigenvectors of $\frac{1}{2} M_{p}(\boldsymbol L^{\text{sym}+}, \boldsymbol Q^{\text{sym}-})$; the result follows. 
\end{proof} 

\paragraph{Signed networks — experiments}
We evaluate on a dynamic model, given in Algorithm~\ref{alg:dynamic-ssbm}, based on the signed stochastic block model found in \cite{cucuringu_sponge_2019, mercado_spectral_2019} and the standard dynamic stochastic block model from Section~\ref{sec:experiments}. The results are shown in Figure \ref{fig:signed-results}.

\begin{algorithm}
\caption{Dynamic Signed Stochastic Block Model (Dynamic SSBM)}
\label{alg:dynamic-ssbm}
\begin{algorithmic}[1]
\Require Probabilities $p_{\text{in}}$, $p_{\text{out}}$, $\eta_{\text{in}}$, $\eta_{\text{out}}$, $p_{\text{switch}}$, number $k$ of planted communities
\State \textbf{At time $i=1$:}
\State Partition node set (up to remainder) into $k$ equally-sized planted communities. Let $Z(\ell)$ denote the community to which node $\ell$ belongs.
\State For every node pair $(i,j)$, place a positive edge ($+1$) between $i$ and $j$ with probability $p_{\text{in}}$ if $Z(i)=Z(j)$, and a negative edge $(-1)$ with probability $p_{\text{out}}$ if $Z(i) \neq Z(j)$. Note that, unlike the unsigned case, it is sensible to choose $p_{\text{in}}=p_{\text{out}}$. 
\State \textit{Flip} the sign of each placed edge according to respective probabilities $\eta_{\text{in}}$ and $\eta_{\text{out}}$ for $\eta_{\text{in}}, \eta_{\text{out}}< \frac{1}{2}$. Commonly $\eta_{\text{in}}=\eta_{\text{out}}$. 
\State \textbf{Repeat the above for snapshot indices }$2 \leq i \leq T$. Dynamics are induced by asserting that, at a given time step $i$, each node (that has not before switched) switches community with probability $p_{\text{switch}}$, its destination chosen uniformly at random among the $k-1$ options. 
\end{algorithmic}
\end{algorithm}

\begin{figure}[htbp]
    \centering
    \begin{subfigure}[b]{0.45\textwidth}
        \centering
        \includegraphics{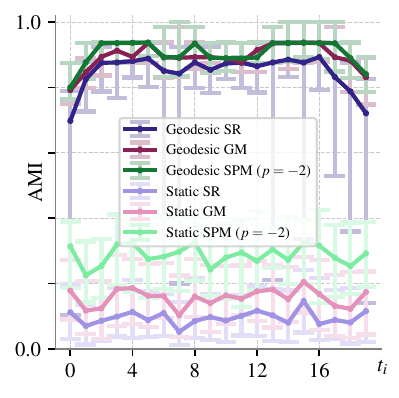}
        \caption{$p_{\text{in}} > p_{\text{out}}$}
        \label{subfig:ssbm-pin-greater-pout}
    \end{subfigure}
    \hfill
    \begin{subfigure}[b]{0.45\textwidth}
        \centering
        \includegraphics{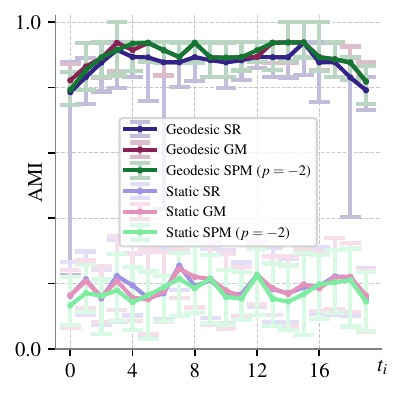}
        \caption{$p_{\text{in}} = p_{\text{out}}$}
        \label{subfig:ssbm-pin-equal-pout}
    \end{subfigure}
    \caption{Comparison of geodesic and static signed community detection methods, medianed over $50$ simulations of two settings of a dynamic signed stochastic block model ($d=120$, $T=20$, $k=2$, $p_{\text{switch}}=10^{-2}$, $\eta_{\text{in}}=\eta_{\text{out}}=0.4$). (\ref{subfig:ssbm-pin-greater-pout}): $p_{\text{in}}=0.3$, $p_\text{out}=0.2$. (\ref{subfig:ssbm-pin-equal-pout}): $p_{\text{in}}=p_{\text{out}}=0.3$: even when the intra- and inter-community connection probabilities are identical, the algorithms are still able to discover community structure based solely on positive and negative edge affinities. We compare algorithms based on the signed ratio Laplacian (SR) \cite{kunegis_spectral_2010}, the geometric mean Laplacian (GM) \cite{mercado_clustering_2016}, and the more general matrix power mean Laplacian (SPM) \cite{mercado_spectral_2019} with $p=-2$.}
    \label{fig:signed-results}
\end{figure}

\subsection{Mixed-Membership Networks}
\label{sec:mixed-membership}
Oftentimes a node belongs to more than one community, e.g., in many social \cite{goldberg_finding_2010, reid_partitioning_2013} and neuronal \cite{kim_detecting_2015} networks. The output of an overlapping community detection algorithm generally takes the form of a $d \times k$ \textit{membership matrix} whose $j\ell$th element equals the estimated probability that node $j$ belongs to community $\ell$. For the purposes of evaluation, this matrix is often thresholded into a binary matrix whose $j$th row indicates the communities to which the $j$th node belongs. 

\paragraph{Mixed-membership networks — algorithms and modeled clustering matrices} 

A simple extension of spectral clustering to the context of overlapping communities is to replace the $k$-means step with a fuzzy $c$-means step. See \cite{nayak_fuzzy_2015} for a discussion on fuzzy $c$-means and \cite{wahl_hierarchical_2015} for analysis of the `fuzzy spectral clustering' method it induces. Other mixed-membership spectral methods include \cite{zhang_detecting_2020}, which compute a thin spectral decomposition $\boldsymbol A=\boldsymbol V_{k} \boldsymbol \Lambda_{k} \boldsymbol V_{k}^{\top}$ (eigenvalues in descending order), and defines node embeddings as $\boldsymbol X=\boldsymbol V_{k} \boldsymbol \Lambda_{k}^{1/2}$. After normalization and regularization, the authors apply $k$-medians clustering (as analyzed in \cite{arora_approximation_1998}); projecting the rows of $\boldsymbol X$ onto the subspace spanned by cluster centers yields a final $d \times k$ matrix representing soft cluster memberships. 

\paragraph{Mixed-Membership Networks — Experiments}
We evaluate on a dynamic model based on the mixed-membership stochastic block model described in \cite{airoldi_mixed_2008}.
To understand the static mixed-membership stochastic block model, we will provide a new interpretation of the (single-membership) stochastic block model, and then generalize it to the mixed-membership case. Following this, we will extend to the dynamic case. 

One can view the connectivity under the version of the stochastic block model from Section~\ref{sec:experiments} in terms of the parameter matrix \begin{equation}
    \begin{bmatrix}
p_{\text{in} }  & p_{\text{out}}  & \dots  & p_{\text{out}} \\
p_{\text{out} }  & p_{\text{in}}  & \dots& p_{\text{out}} \\
\vdots  & \vdots & \ddots & \vdots  \\
p_{\text{out}}  & p_{\text{out}}  & \dots  & p_{\text{in}} 
\end{bmatrix}=: \boldsymbol  B \in \mathbb{R}^{k \times k}
\end{equation}
where each element of is a Bernoulli random variable with parameter $p_{\text{in}}$ or $p_{\text{out}}$. Then, with $Z(\ell)$ denoting the planted community of a node $\ell$, the probability $\mathbb{P}(A_{ij}=1)$ of an edge between two nodes $i,j \in [d]$ is $B_{Z(i),Z(j)}$, which we can write as  $$\mathbb{P}(A_{ij}=1) = B_{Z(i),Z(j)} = \phi_{i,Z(i)}\phi_{j,Z(j)}B_{_{Z(i),Z(j)}} = \sum_{g=1}^{k}\sum_{h=1}^{k}   \phi_{ig}     \phi_{jh} B_{gh},$$
where $\boldsymbol \phi_{\ell}=[0  \cdots   0  \overbrace{1}^{g\text{th entry}}  0  \cdots  0 ]^{\top} \in \mathbb{R}^{k \times k}$ indicates that node $\ell$ belongs to the $g$th planted community with probability 1. 
A mixed-membership prescription then follows by allowing $\boldsymbol \phi_{\ell}$ to not be an indicator vector, but rather a normalized vector of probabilities wherein $\phi_{\ell g}$ describes the probability that node $\ell$ belongs to the $g$th community. We introduce simple dynamics into the model by asserting that, at each snapshot index, a node switches with probability $p_{\text{switch}}$, with its destination chosen uniformly at random from among all planted communities and their intersections. See Algorithm~\ref{alg:dynamic-mmsbm}.

We evaluate using \new{element-centric similarity as outlined in \cite{gates2019element}}. Although the overlapping community detection algorithms described above output community membership probabilities for each node, the \new{element-centric similarity metric utilizes binary assignments}. To be compatible, we threshold such that a node is declared part of a community if it belongs to that community with probably exceeding $p_{\text{thresh}}:=0.2$. Results are in Figure \ref{fig:mmsbm-results}.


\begin{algorithm}[H] 
\caption{Dynamic Mixed-Membership Stochastic Block Model}
\label{alg:dynamic-mmsbm}
\begin{algorithmic}[1]
\Require Number of communities $k$, number of nodes $d$, number of time points $T$, parameter matrix $\boldsymbol B$, switching probability $p_{\text{switch}}$, initial mixed-membership vectors $\bs \Phi \in \mathbb{R}^{d \times k}$ with columns $\{\boldsymbol \phi_{\ell}\}_{\ell=1}^d \subset \mathbb{R}^k$
\State \textbf{At time $i=1$:}
\State Use provided $\boldsymbol \phi_{\ell} \in \mathbb{R}^k$ for each node $\ell \in [d]$
\State Sample an adjacency matrix $A$,  $\mathbb{P}(A_{ij}=1) = \sum_{g=1}^{k}\sum_{h=1}^{k} \phi_{ig} \phi_{jh} B_{gh}$
\State \textbf{Repeat the above for $2 \leq i \leq T$}. Dynamics are induced as follows: at a given time step $i$, each node switches mixed-membership vector with probability $p_{\text{switch}}$, its new vector chosen uniformly at random from among the unique columns of $\bs \Phi$.
\end{algorithmic}
\end{algorithm}

\begin{SCfigure}
    \includegraphics{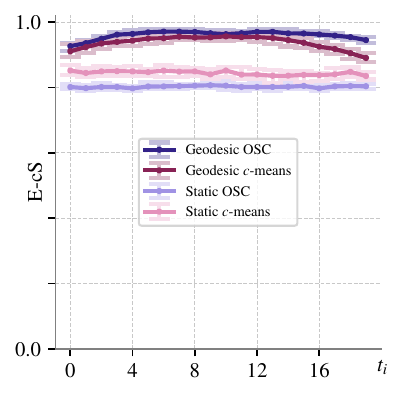}
    \caption{Comparison of geodesic and static overlapping community detection methods in time, medianed over $50$ simulations of a dynamic mixed-membership stochastic block model ($d=120$, $k=2$, $T=20$, $p_{\text{switch}}=10^{-2}$, $p_\text{in}=0.3$, $p_{\text{out}}=0.2$). 
    At time $i=1$, two communities are planted, each possessing $50$ `single-affiliated' nodes. The remaining $20$ nodes are equally affiliated with each community. That is, $\boldsymbol{\phi}_1=\begin{bmatrix}1  &  \cdots & 1  & 0.5  & \cdots  & 0.5  & 0  & \cdots  &  0 \end{bmatrix}^{\top}$ and  $\boldsymbol \phi_{2}=\begin{bmatrix} 0  & \cdots  &  0 & 0.5  & \cdots  & 0.5  & 1  &  \cdots & 1 \end{bmatrix}^{\top}$. Error bars correspond to $25$th and $75$th percentiles. The algorithms compared are the overlapping spectral clustering (OSC) method of \cite{zhang_detecting_2020} and  $c$-means spectral clustering based on \cite{wahl_hierarchical_2015}.} 
    \label{fig:mmsbm-results}
\end{SCfigure}

\subsection{Directed Networks}
\label{sec:directed-networks}
A great deal of network modalities lend themselves to asymmetric representations, including many social \cite{wang_detecting_2011}, informational \cite{huang_web_2006}, biological \cite{guimera_origin_2010, popa_directed_2011} and neuroscientific \cite{deco_dynamical_2011} networks. We consider two notions of `community' which naturally arise in this directed network context. The first aligns with the undirected case: a good partition divides the network into communities such that intracommunity edge densities far exceed intercommunity edge densities. The second is especially pertinent to the directed case, and concerns the grouping of nodes based on `patterns' rather than edge density. These include co-citation patterns (outgoing resp. incoming edges on in-community nodes are more likely to share common targets resp. sources) and flow-based patterns (a random walker is more likely to get `stuck' in a community) \cite{malliaros_clustering_2013}. 

\paragraph{Directed Networks — algorithms and modeled clustering matrices} One popular extension of spectral methods to the directed context is based on replacement of the directed adjacency matrix $\boldsymbol{A}$ with an appropriate symmetrization.  Two such symmetrizations are evaluated in \cite{satuluri_symmetrizations_2011}  in the context of spectral clustering. The first is \textit{bibliographic symmetrization}:\begin{equation}\boldsymbol{A}^{\text{bibliographic}} =\boldsymbol A \boldsymbol A^{\top} +  \boldsymbol A^{\top} \boldsymbol A,
    \label{bibliographic-sym}
\end{equation} where the \textit{bibliographic coupling matrix} $\boldsymbol{A}\boldsymbol{A}^\top$ counts the number of nodes in $G$ to which $i$ and $j$ both point \cite{kessler_bibliographic_1963} and the \textit{co-citation matrix} $\boldsymbol{A}^\top \boldsymbol{A}$ counts the number of nodes in $G$ that point to both $i$ and $j$  \cite{small_co_1973}. The second is  \textit{degree-discounted symmetrization}: \begin{align}
    \boldsymbol A^{\text{degree-discounted}}:=\boldsymbol  D_{\text{out}}^{-1/2}\boldsymbol  A \boldsymbol  D_{\text{in}}^{-1/2}\boldsymbol  A^{\top}\boldsymbol  D_{\text{out}}^{-1/2} + \boldsymbol  D_{\text{in}}^{-1/2}\boldsymbol  A^{\top}\boldsymbol  D_{\text{out}}^{-1/2}\boldsymbol  A \boldsymbol  D_{\text{in}}^{-1/2},
\end{align} which normalizes (\ref{bibliographic-sym}) to account for the heterogeneity of degree distributions found in real-world networks \cite{satuluri_symmetrizations_2011}. Also studied are Laplacian extensions for directed graphs, the most popular among these \cite{chung_laplacians_2005, zhou_learning_2005, gleich_hierarchical_2006, malliaros_clustering_2013} being the random walk-based directed Laplacian (RW) whose \textit{largest} eigenvectors encode graph partition structure \begin{equation}
    \boldsymbol \Theta=
    \frac{1}{2}\left( \boldsymbol \Pi^{1/2} \boldsymbol P \boldsymbol \Pi^{-1/2} + \boldsymbol \Pi^{-1/2} \boldsymbol P^{\top} \boldsymbol \Pi^{1/2} \right),
\end{equation}
 where $\boldsymbol P$ is the transition matrix $P_{ij}=\frac{A_{ij}}{\sum_{j=1}^{d}A_{ij}}$ and $\boldsymbol \Pi=\boldsymbol D_{\text{out}}^{-1}$.

We note that the aforementioned symmetrization approaches are intended to detect both density-based and pattern-based communities, whereas the directed Laplacian-based approaches are intended only to detect density-based communities (\cite{malliaros_clustering_2013}, Table 2). 

\paragraph{Directed networks — experiments}
We evaluate on following dynamic model (Algorithm~\ref{alg:dynamic-dsbm}), based on the directed stochastic block model proposed in \cite{cucuringu_hermitian_2020}.  Results are shown in Figure~\ref{fig:dsbm-results}.

\begin{algorithm}
\caption{Dynamic Directed Stochastic Block Model}
\label{alg:dynamic-dsbm}
\begin{algorithmic}[1]
\Require probabilities $p_{\text{in}}$, $p_{\text{out}}$, $p_{\text{switch}}$, matrix $\boldsymbol F \in [0,1]^{k \times k}$ satisfying $F_{\ell j}+F_{j\ell}=1$ for all $j,\ell \in [k]$, number $k$ of planted communities
\State \textbf{At time $i=1$:}
\State Partition node set (up to remainder) into $k$ equally-sized planted communities. Let $Z(\ell)$ denote the community to which the node $\ell$ belongs
\State For every node pair $(i,j)$, place an edge between $i$ and $j$ with probability $p_{\text{in}}$ if $Z(i)=Z(j)$ and probability $p_{\text{out}}$ if $Z(i) \neq Z(j)$. The edge is directed from $i$ to $j$ with probability $F_{Z(i), Z(j)}$; otherwise, it is directed from $j$ to $i$. 
\State \textbf{Repeat the above for snapshot indices $2 \leq i \leq T$}. Dynamics are induced by asserting that, at a given time step $i$, each node (that has not before switched) switches community with probability $p_{\text{switch}}$, its destination chosen uniformly at random among the $k-1$ options.
\end{algorithmic}
\end{algorithm}
\begin{figure}
    \centering
    \begin{subfigure}[b]{0.48\textwidth}
        \centering
        \includegraphics[width=\textwidth]{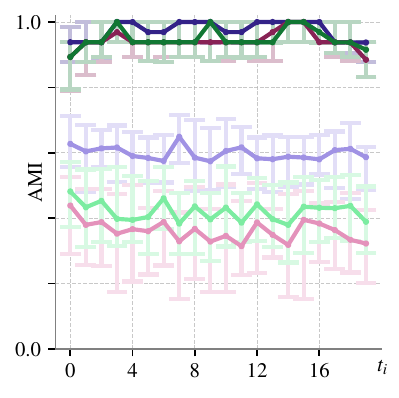}
        \caption{Density-based clustering}
        \label{subfig:density}
    \end{subfigure}
    \hfill
    \begin{subfigure}[b]{0.48\textwidth}
        \centering
        \includegraphics[width=\textwidth]{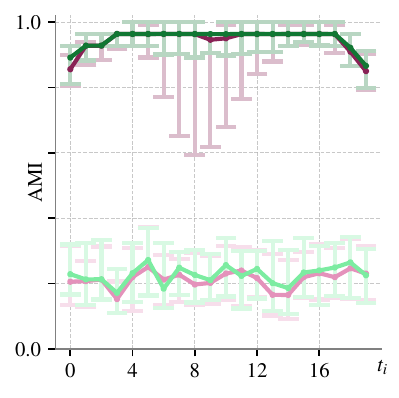}
        \caption{Pattern-based clustering}
        \label{subfig:pattern}
    \end{subfigure}
    
    \vspace{1em}
    
    \begin{subfigure}[b]{0.3\textwidth}
        \centering
        \includegraphics[width=\textwidth]{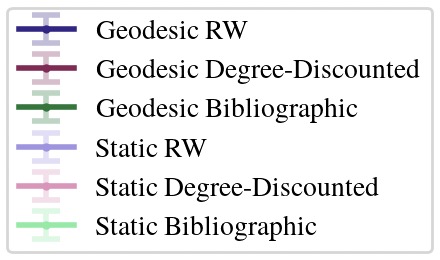}
    \end{subfigure}
    \hfill
    \begin{minipage}[b]{0.68\textwidth}
        \centering
        \caption{Comparison of geodesic and static methods medianed over $50$ simulations of two settings of a dynamic directed stochastic block model, both with $d=120$, $T=20$, $p_{\text{switch}}=10^{-2}$. (\ref{subfig:density}) $k=2$ communities are planted with $\boldsymbol F=\begin{bmatrix}0.5  & 0.4  ; 0.6 & 0.5\end{bmatrix}$, $p_{\text{in}}=0.3$, $p_\text{out}=0.2$ in a `density-based' parameter setting. (\ref{subfig:pattern}) $k=3$ communities are planted with $\boldsymbol{F}=\begin{bmatrix}1 / 2  & 2 /3  & 1 / 3  \\ 1 / 3  & 1 / 2 & 2 / 3 \\2 / 3 &  1/3  & 1/2\end{bmatrix}$, $p_{\text{in}}=p_{\text{out}}=0.2$ in a `flow/pattern-based' parameter setting \cite{cucuringu_hermitian_2020}. 
        }
        \label{fig:dsbm-results}
    \end{minipage}
    
\end{figure}

\subsection{Networks with Cocommunity or Bipartite Structure}
\label{sec:cocommunity-networks}
An adjacent problem to that of directed network community detection is that of co-community detection. The general coclustering problem concerns the grouping of data according to multiple attributes (e.g., samples and features) simultaneously. We  focus on two notions of cocommunity detection found in network analysis. The first notion applies to directed networks, where the two attributes correspond to rows and columns of the directed adjacency matrix. The output is two network partitions: one which clusters nodes with similar sending patterns, and another which clusters nodes with similar receiving patterns. The second notion applies to bipartite networks, where the goal is to cluster nodes of both types simultaneously. The output is two coindexed partitions of first-type nodes and second-type nodes respectively. The cocommunity detection method we generalize via Algorithm~\ref{alg:geodesic-dcd} applies to both notions.

\paragraph{Algorithms and modeled clustering matrices} We will consider the spectral coclustering method found in \cite{rohe_co_2016}. A related algorithm may be found in \cite{dhillon_co_2001}. Given a numbers $k_y$ and $k_z$ of sending clusters and receiving clusters to detect, the method computes a $k$-truncated singular value decomposition $\boldsymbol{U} \boldsymbol\Sigma \boldsymbol{V}^\top$ of a regularized  (asymmetric) graph Laplacian $\boldsymbol{L}$, where $k=\min(k_y, k_z)$. $\boldsymbol{U}$ is then used for `sending embeddings', while $\boldsymbol{V}$ is used for `receiving embeddings'. From the perspective of Algorithm~\ref{alg:static-scd-template}, $\boldsymbol{L}$ is already in MCM form when the goal is to detect receiving clusters. When the goal is to detect sending clusters, $\boldsymbol{L}^\top$ has the requisite form. 


\paragraph{Experiments} We evaluate with a dynamic model, given in Algorithm \ref{alg:coblock}, based on the dynamic stochastic coblock model from \cite{rohe_co_2016}. Results are given in Figure \ref{fig:scc-results}.

\begin{algorithm}
\caption{Dynamic Stochastic Coblock Model}
\label{alg:coblock}
\begin{algorithmic}[1]
\Require Number $k_y$ of sending communities, number $k_z$ of receiving communities; $\boldsymbol B \in \mathbb{R}^{k_y \times k_z}$ where $B_{\ell \ell'}=p_{\ell \ell'} \in [0,1]$ represents the probability of an edge existing from a node in sending cluster $\ell \in [k_y]$ to a node in receiving cluster $\ell' \in [k_z]$. Switching probabilities $p_{\text{switch, send}}$ and $p_{\text{switch, receive}}$.
\State If the graph is bipartite with node sets $V_1$ and $V_2$, declare $V_1$ the sending node set and $V_2$ the receiving node set. If the graph is not bipartite, both sending and receiving node sets are the full node set $V=V_{1}=V_{2}$.
\State Assign each node in the sending node set $V_{1}$ to one of the $k_y$ sending communities, divided equally. Similarly, assign each node in the receiving node set to one of the $k_z$ receiving communities, divided equally.
\State For each node pair $(i,j)$, place a (directed) edge from node $i$ to node $j$ with probability $B_{y_i,z_j}$, where $y_i$ is the sending community of $i$ and $z_j$ is the receiving community of $j$.
\State \textbf{Repeat for time steps $2 \leq i \leq T$:}
\State \quad At each time step $i$, for each node in $V_1$, with probability $p_{\text{switch, send}}$, swap its sending community assignment with a randomly chosen node from a uniformly random different sending community.
\State \quad Similarly, for each node in $V_2$, with probability $p_{\text{switch, receive}}$, swap its receiving community assignment with a randomly chosen node from a uniformly random different receiving community.
\State \quad After community assignments are updated, regenerate all edges according to step 3.
\end{algorithmic}
\end{algorithm}

\begin{figure}[htbp]
    \centering
    \begin{subfigure}[b]{0.45\textwidth}
        \centering
        \includegraphics{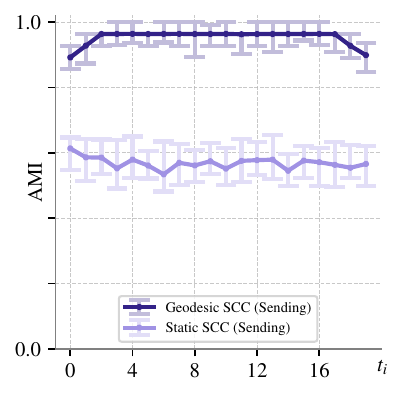}
        \label{subfig:scbm-sending}
    \end{subfigure}
    \hfill
    \begin{subfigure}[b]{0.45\textwidth}
        \centering
        \includegraphics{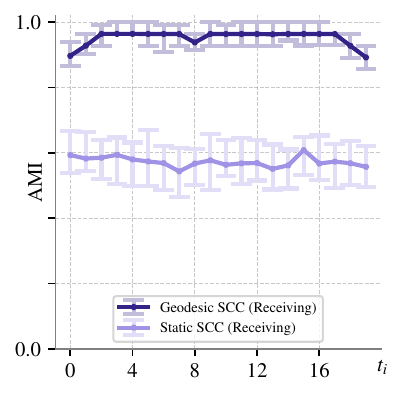}
        \label{subfig:scbm-receiving}
    \end{subfigure}
    \caption{Comparison of geodesic and static versions of the spectral coclustering (SCC) method in \cite{rohe_co_2016} over $50$ simulations of the dynamic stochastic coblock model ($d=120$, $T=20$, $p_{\text{switch, send}}=p_{\text{receive, send}}=10^{-2}$, $k_y=k_z=3$, $\boldsymbol  B=\begin{bmatrix}0.5  & 0.3  & 0.3 \\0.3  & 0.5  & 0.3\\0.3  & 0.3 & 0.5\end{bmatrix}$). Ribbons correspond to $25$th and $75$th percentiles.}
    \label{fig:scc-results}
\end{figure}

\subsection{Hierarchical Networks}
\label{sec:hierarchical-networks}
A great deal of real-world complex systems exhibit interesting behavior at multiple resolutions. Examples of detecting hierarchical communities within such systems abound in neuroscience \cite{ashourvan_multi_2019}, biochemistry \cite{ravasz_hierarchical_2002}, and the social sciences \cite{rezvani_survey_2022}. The output of a hierarchical community detection algorithm is a rooted tree whose vertices represent communities, edges represent parent-child relationships, and levels represent scales of resolution. 

\paragraph{Hierarchical networks — Algorithms and modeled clustering matrices}
We explore the approach to hierachical community detection described in \cite{laenen_nearly_2023}, which uses Laplacian eigenvectors to obtain the coarsest partition, then employs a degree-bucketing approach to unfold the community hierarchy. Another spectral method for hierarchical community detection (not included in our experiments) may be found in \cite{schaub_hierarchical_2023}.

\paragraph{Hierarchical networks — experiments}
Evaluation is performed on a dynamic model, given in Algorithm \ref{alg:dynamic-hsbm}, based on the hierarchical stochastic block model proposed in \cite{cohen_hierarchical_2019}. We score according to the hierarchial normalized mutual information metric proposed in \cite{perotti_hierarchical_2015}. 
Results are found in Figure \ref{fig:three-subfigs}.

\begin{algorithm}[H]
\caption{Dynamic Hierarchical Stochastic Block Model}
\label{alg:dynamic-hsbm}
\begin{algorithmic}[1]
\Require A rooted tree on $L$ leaves consisting of $M$ nodes with weights $p_{m} \in [0,1]$, $m \in [M]$, wherein each node corresponds to a planted community, and each level of the tree corresponds to a level of hierarchy in the network. Note that the weight assigned to the root node coincides with $p_{\text{out}}$ in the stochastic block model described in Section~\ref{sec:experiments}. Number $T$ of time steps. Probability $p_\text{switch}$. 
\State \textbf{At time $i=1$:}
\State Assign $d$ nodes equally (up to remainder) among the $L$ leaves
\State For every pair of nodes $i$, $j$ belonging to leaves $L(i)$, $L(j)$ respectively, place an edge with probability $p_{m}$, where $m$ denotes the least common ancestor of $L(i)$ and $L(j)$ 
\State \textbf{For time steps $2 \leq i \leq T$, repeat the above.} Dynamics are induced as follows:  at a given time step $i$, each node (that has not before switched) switches leaf community to one of its siblings with probability $p_{\text{switch}}$, its destination chosen uniformly at random.
\end{algorithmic}
\end{algorithm}

\begin{figure}
    \centering
    \begin{subfigure}[t]{0.3\textwidth}
        \centering
        \includegraphics{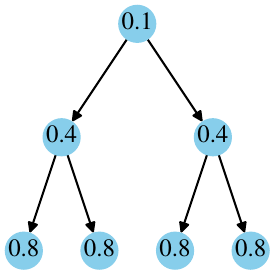}
        \label{fig:hierarchical-tree-example}
    \end{subfigure}
    \hfill
    \begin{subfigure}[t]{0.3\textwidth}
        \centering
        \includegraphics{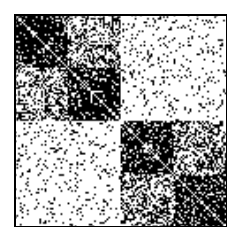}
        \label{fig:adj-matrix-example}
    \end{subfigure}
    \hfill
    \begin{subfigure}[t]{0.3\textwidth}
        \centering
        \includegraphics{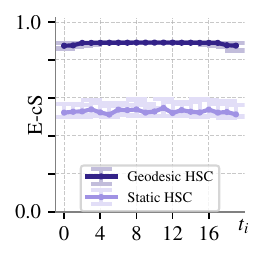}
        \label{fig:hsbm-results}
    \end{subfigure}
    \caption{LEFT: An example probability tree input into the dynamic hierarchical stochastic block model. MIDDLE: An adjacency matrix, sampled from the dynamic HSBM according to this tree. RIGHT: Comparison of geodesic and static versions of the hierarchical spectral clustering (HSC) method in \cite{laenen_nearly_2023}, medianed over $50$ simulations of the dynamic HSBM ($d=120$, $T=20$, $p_{\text{switch}}=10^{-2}$, probability tree mimicking that displayed but with values in $0.4 \pm 0.05$ to introduce noise.) Error bars correspond to $25$th and $75$th percentiles.}  
    \label{fig:three-subfigs}
\end{figure}

\subsection{Multiview Networks}
\label{sec:multiview-networks}
In practice, one often has access to multiple graphs corresponding to the same network. Experiments used to construct network representations of data often have many trials \cite{newman_network_2018}. Or there may exist many modes of relationship between the same set of nodes, each yielding its own `view' of the network — \cite{magnani_community_2021}. Like snapshot-represented dynamic networks, the underlying representation of a multiview network is generally a multiplex graph. Unlike dynamic networks, the goal is not to detect one community partition per layer, but rather to construct a single partition of the nodes using information from multiple layers. That said, the two problem settings can interact. For instance, suppose the observed dynamic simple graph $\{G_i\}_{i=1}^T$ of Section~\ref{sec:simple-networks} is of high temporal resolution, in the sense that $T$ is very large but the network evolution between adjacent snapshots is minuscule, perhaps with the exception of some outliers of high temporal discontinuity. By segmenting the snapshots into `windows', a dynamic multiview network can be created, and applying a dynamic multiview community detection to this network might capture a coarsened perspective on the unfolding community structure, perhaps while `smoothening away' the outliers. 

\paragraph{Multiview networks — algorithms and modeled clustering matrices} 
We discuss the dynamic generalization of two spectral methods for detecting communities in multiview networks, each based on analyzing the spectrum of a single `summary Laplacian' computed using the graph Laplacians from the individual layer. The first approach is directly adjacent to the signed power mean Laplacian method discussed in Section~\ref{sec:signed-networks}, where the spectrum of the power mean of $S$ Laplacians — one per view of the network — is considered \cite{mercado_power_2018}. The second approach is based on Grassmann manifold geometry \cite{dong_clustering_2013}: the authors look at Laplacian spectrums layer-by-layer to obtain a collection of subspaces $\{ \langle  \bs U_{i} \rangle\}_{i=1}^{S}$, then solve a Riemannian optimization problem on the Grassmann manifold for combining the $\langle \bs U_{i} \rangle$ into a single representative subspace $\langle \bs U \rangle$ that is `close' to each $\langle \bs U_{i} \rangle$. They obtain from this analysis a single matrix whose spectrum aims to summarize the clustering structure across all layers.

\paragraph{Multiview networks — experiments}
We evaluate with a dynamic model, given in Algorithm \ref{alg:dynamic-mvsbm}, based on a simple setting of the multiview stochastic block model analyzed in \cite{zhang_community_2024}. 
Results are found in Figure \ref{fig:mvsbm-results}. 

\begin{algorithm}[H] 
\caption{Dynamic Multiview Stochastic Block Model}
\label{alg:dynamic-mvsbm}
\begin{algorithmic}[1]
\Require Number $S$ of views, $T$ time points, $p_{\text{in}}$, $p_{\text{out}}$, $p_{\text{switch}}$
\State \textbf{At time $i=1$:}
\State Partition node set (up to remainder) into $k$ equally-sized planted communities. Let $Z(\ell)$ denote the community to which node $\ell$ belongs.
\State Sample $S$ adjacency matrices from the (static) stochastic block model described in Section~\ref{sec:experiments} according to $p_{\text{in}}$ and $p_{\text{out}}$ and concatenate into an $S \times d \times d$ adjacency tensor
\State \textbf{Repeat the above for snapshot indices $2 \leq i \leq T$}. Dynamics are induced as follows: at a given time step $i$, each node (that has not before switched) switches community with probability $p_{\text{switch}}$, its destination chosen uniformly at random among the $k-1$ options.
\end{algorithmic}
\end{algorithm}

\begin{SCfigure}
    \includegraphics{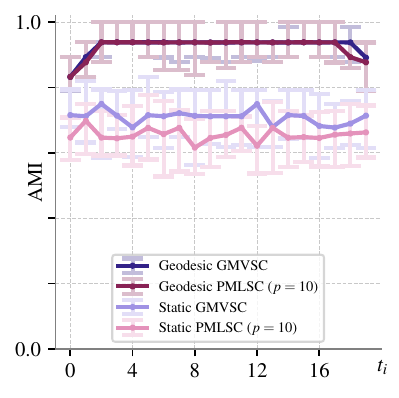}
    \caption{Comparison of geodesic and static multiview community detection methods in time, medianed over $50$ simulations of a dynamic mixed-membership stochastic block model ($d=120$, $k=2$, $T=20$, $p_{\text{switch}}=10^{-2}$, $p_\text{in}=0.3$, $p_{\text{out}}=0.2$, $S=3$). Error bars correspond to $25$th and $75$th percentiles. The algorithms compared are the Grassmannian multiview spectral clustering (GMVSC) algorithm of \cite{dong_clustering_2013} and the power mean Laplacian spectral clustering (PMLSC) method of \cite{mercado_spectral_2019} with $p=10$.} 
    \label{fig:mvsbm-results}
\end{SCfigure}

\subsection{Other Network Modalities}
\label{sec:other-networks}
Many successful methods for community detection in principle admit generalization via Algorithm~\ref{alg:geodesic-dcd}, even though they are not included among our experiments. Examples include spectral methods for hypergraph clustering \cite{zhou_learning_2006, concas_spectral_2020},  motif-based methods \cite{underwood_motif_2020} and higher-order, topological methods \cite{grande_disentangling_2024, krishnagopal_spectral_2021}. We also note that many of the non-regularized methods explored in this text have regularized counterparts which generalize via our method to the evolving setting as well \cite{qin_regularized_2013, cucuringu_regularized_2021, qing_regularized_2023, kumar_co_2011}.

\section{Proof of Proposition~\ref{prop:geodesic-assumption}}
\label{sec:appendix-geodesic-assumption}
    
    Suppose $\sigma_{1} \geq \sigma_{2} > \sigma_{3}=\dots=\sigma_{m}=0$, so that $\boldsymbol X=\boldsymbol a_{1} \sigma_{1}\boldsymbol b_{1}^{\top} + \boldsymbol a_{2} \sigma_{2}\boldsymbol b_{2}^{\top}$ and $\langle\boldsymbol X \rangle=\operatorname{span}(\boldsymbol a_{1},\boldsymbol a_{2})=\mathcal{U}$, where $\mathcal{U}$ is a 2-plane in $\mathbb{R}^{d}$ containing the origin. It follows that for each column $\boldsymbol c$ of $\boldsymbol X$ we have $$\boldsymbol  c \in \mathbb{S}^{d-1} \cap \mathcal{U},$$
so that $\boldsymbol c$ belongs to a great circle arc in $\mathbb{S}^{d-1}$. Hence the columns of $\boldsymbol X$ each live on the trace of a geodesic $\gamma:[0,1] \to  \mathbb{S}^{d-1}$. It suffices, then, to show that $\varphi\circ \gamma$ is a geodesic on the real projective space $\text{Gr}(1,d)$, where $\varphi:\mathbb{S}^{d-1} \to \text{Gr}(1,d)$ is given by $\varphi(\boldsymbol w):=\langle \boldsymbol w \rangle$. 

To this end, let the Lie group $\mathbb{Z} / 2\mathbb{Z} = \{1,-1\}$ act on $\mathbb{S}^{d-1}$ by (left) multiplication. The orbits of this action induce an equivalence relation (\cite{boumal_introduction_2023}, Definition 9.13) whose associated projection $\pi: \mathbb{S}^{d-1} \to \mathbb{S}^{d-1} / (\mathbb{Z} / 2\mathbb{Z})$ sends $\boldsymbol w \in \mathbb{S}^{d-1}$ to $[\boldsymbol w]=\{ \boldsymbol w, - \boldsymbol w \}$. This action is clearly smooth, free, and — since $\mathbb{Z}/2\mathbb{Z}$ is compact as a finite topological group — it is proper (\cite{lee_introduction_nodate}, Corollary 21.6). It is a local isometry as well; checking this amounts to verifying (\cite{boumal_introduction_2023}, Theorem 9.38) that the antipodal map $f:\mathbb{S}^{d-1} \to \mathbb{S}^{d-1}$, $\boldsymbol w \mapsto -\boldsymbol w$ satisfies, for all $\boldsymbol w \in \mathbb{S}^{d-1}$ and $\boldsymbol x_{1}, \boldsymbol x_{2} \in T_{\boldsymbol w}(\mathbb{S}^{d-1})$, $$\langle df(\boldsymbol w )[\boldsymbol  x_{1}], df(\boldsymbol  w)[\boldsymbol  x_{2}] \rangle_{f(\boldsymbol  w)} = \langle \boldsymbol  x_{1}, \boldsymbol  x_{2} \rangle_{\boldsymbol  w}.  $$ Because $df(\boldsymbol w)=-\text{id}_{\mathbb{S}^{d-1}}$, we can rewrite the desired assertion as $$\langle -\boldsymbol  x_{1}, -\boldsymbol  x_{2}  \rangle_{-\boldsymbol  w} = \langle \boldsymbol  x_{1}, \boldsymbol  x_{2} \rangle_{- \boldsymbol  w} =\langle \boldsymbol  x_{1}, \boldsymbol  x_{2} \rangle _{\boldsymbol  w};$$
since the Riemannian metric on $\mathbb{S}^{d-1}$ is inherited from the Euclidean inner product on $\mathbb{R}^{d}$, and the latter is invariant under orthogonal transformations, the equation holds. $\pi$ is a therefore a Riemannian covering map (\cite{gallier_differential_2020}, Theorem 23.18), and hence (\cite{gallier_differential_2020}, Proposition 18.6) it projects and lifts geodesics to geodesics.

Moreover, recall that the canonical geodesic distance between two points on the Grassmannian is $\| \Theta\|_{2}$, where $\Theta \in \mathbb{R}^k$ consists of principal angles between the subspaces \cite{bendokat_grassmann_2024}. It follows immediately that the (clearly well-defined) diffeomorphism $\psi:\mathbb{S}^{d-1} / (\mathbb{Z} / 2\mathbb{Z}) \to \text{Gr}(1,d)$ mapping $[\boldsymbol w]$ to $\langle \boldsymbol w\rangle$ preserves geodesic distances.  

Thus, $\varphi$ factors as $\psi \circ \pi$, where $\psi$ and $\pi$ each preserve geodesics. It follows that $$\{\varphi( \pm \boldsymbol  v_{1}^{(1)}),\dots, \varphi(\pm \boldsymbol  v_{1}^{(d)})\} =\{ \langle \pm\boldsymbol  v_{1}^{(1)}\rangle ,\dots, \langle\pm \boldsymbol  v_{1}^{(d)}\rangle\}$$
belongs to the trace of a Grassmann geodesic, as claimed. 

Conversely, suppose the first singular subspaces $\langle \boldsymbol v_{1}^{(1)} \rangle,\dots,\langle \boldsymbol v_{1}^{(T)} \rangle$ lie on the trace of a geodesic $\delta:[0,1] \to \text{Gr}(1,d)$. Since $\psi$ is a diffeomorphism that preserves geodesic distances, $\psi ^{-1}$ is as well. Additionally, $\pi$ lifts geodesics to geodesics as a Riemannian covering map. By lifting $\delta$ through $\psi$ and $\pi$ to geodesics on $\mathbb{S}^{d-1}$, we can realize each of the $\pm \boldsymbol v^{(i)}_{1}$ as living on the trace of a geodesic on $\mathbb{S}^{d-1}$, i.e., as living in the intersection of $\mathbb{S}^{d-1}$ with a 2-dimensional linear subspace $\mathcal{U}$. Hence $\sigma_{3}=\dots=\sigma_{m}=0$. 

\paragraph{Remark} In spectral clustering when $k_c=2$, one also clusters based on the signs of an eigenvector's entries~\cite{von_tutorial_2007}. However, the most extremal eigenvector is the vector of all ones not interesting (indeed, it corresponds to the partition trivially minimizing the cut size objective by placing all nodes into one cluster\cite{newman_finding_2006}); instead, the second extremal eigenvector $\hat{\bs u}$ is used. All of the discussions phrased here in terms of spectral modularity maximization carry through for spectral clustering provided that one works with $\hat{\bs u}$ instead of the most extremal eigenvector $\bs u$. 
\section{Extension to Time-Varying $k_c$}
\label{sec:variable-k-appendix}
This section elaborates upon the straightforward extension, mentioned in Section~\ref{sec:experiments}, of Algorithm~\ref{alg:geodesic-dcd} to the case where the number of latent communities $k_c$ varies in time. The approach described is applicable to any network modality for which an unsupervised benefit function (e.g. modularity or an extension of it) has been studied.

Section~\ref{sec:variable-k-motivation} will motivate our extension. Section~\ref{sec:variable-k-algorithm} will provide an explicit algorithm for it, modeled off of Algorithm~\ref{alg:geodesic-dcd}, together with empirical results. While the approach taken has some theoretical motivation, it is nevertheless a heuristic, and Section~\ref{sec:variable-k-alternatives} gives potential alternative heuristics and future directions toward a more principled approach. 
\subsection{Motivation}
\label{sec:variable-k-motivation}
Static spectral methods for community detection generally require the number $k_c$ of desired communities, the embedding dimension $k_e$, or both of these values to be specified in advance. The temporal setting amplifies this drawback, since in certain contexts the number of latent communities may earnestly change in time, as a result e.g. of merging or splitting \cite{cazabet_dynamic_2017}. This section justifies an extension of Algorithm~\ref{alg:geodesic-dcd} that is capable of automatically detecting a variable number of communities at each time step. 

Our approach is motivated by the observation that, upon writing the $\operatorname{RatioCut}$ objective function for unnormalized spectral clustering in terms of spectrum $\lambda_1(\bs M), \dots, \lambda_d(\bs M)$ of the MCM $\bs M=n\bs I - \bs L$ (Proposition~\ref{prop:prop-low-rank-unnormalized-sc}), the problem of minimizing $\operatorname{RatioCut}$ becomes equivalent to a flavor of \textit{max-sum vector partitioning} applied to $d$-dimensional spectral embeddings. If we use $k$-dimensional spectral embeddings instead, the vector partitioning problem is approximately equivalent to $\operatorname{RatioCut}$, with error proportional to the energy lost by discarding $\lambda_{k+1}(\bs M), \dots, \lambda_d(\bs M)$. We use this observation to argue that, although letting the spectral embedding dimension $k_e$ equal the desired number of communities $k_c$ is conventional and effective for spectral partitioning, if $k_e$ exceeds $k_c$ by a small-to-moderate amount\footnote{Not by too much, however, in light of noise considerations and dimensionality curses.} then we should still expect good performance (since Euclidean clustering with a larger $k_e$ is theoretically optimizing a function closer to the true $\operatorname{RatioCut}$ objective). The consequence is that, by choosing $k_e$ in Algorithm~\ref{alg:geodesic-dcd} to be an upper bound for the estimated number $k_{c,i}$ of latent communities at any time step $i$, we are permitted to vary $k_c=k_c(t)$ freely as a function of time without penalty. The task then becomes deciding how to automatically choose each $k_{c,i}=k_c(t_i)$; we provide one approach but note that unexplored alternatives may do better (appendix \ref{sec:variable-k-alternatives}).

\paragraph{Minimal graph cuts and maximal vector partitions} The present discussion uses (unnormalized) spectral clustering as a prototypical example, but the argument we offer has analogues in (at least) the contexts of spectral modularity maximization \cite{zhang_multiway_2015} and size-contrained graph partitioning \cite{alpert_spectral_1995}. Said analogues have motivated broader efforts to better understand how large $k_e$ should be compared to $k_c$ in general spectral settings \cite{rebagliati2011spectral}.

Recall the definition of the \textit{cut size objective} for evaluating a partition $(Z_1, \dots, Z_{k_c})$ of a simple graph $G$ with adjacency matrix $\bs A$ into $k_c$ communities: \begin{equation}
    \operatorname{Cut}(Z_{1}, \dots, Z_{k_c}) := \frac{1}{2} \sum_{i=1}^{k} W(Z_{i}, V \setminus Z_i), \text{ where } W(Z, Z') := \sum_{z \in Z, z' \in Z'} A_{zz'},\label{eqn:cut-size-obj}
\end{equation}
and the $\operatorname{RatioCut}$ objective \cite{hagen_new_1992} whose relaxation yields (unnormalized) spectral clustering:\begin{equation}
\label{eqn:ratiocut-obj}
    \operatorname{RatioCut}(Z_{1}, \dots, Z_{k_c}) := \sum_{i=1}^{k_c} \frac{\operatorname{Cut}(Z_{i}, V \setminus Z_i)}{|Z_{i}|}.
\end{equation}

Defining $\boldsymbol S$ to be the $d \times k$ community matrix $S_{ij}=\frac{\mathbbm{1}_{Z(i)=Z(j)}}{\sqrt{ |Z_{j}| }}$, it can be shown \cite{von_tutorial_2007} that $\operatorname{RatioCut}$ may be rewritten \begin{equation}
    \operatorname{RatioCut}(Z_{1},\dots,Z_{k})=\operatorname{Tr}(\boldsymbol  S^{\top} \boldsymbol  L \boldsymbol  S),
\end{equation}
where $\bs L = \bs D - \bs A$ is the unnormalized graph Laplacian of $G$ and $Z(\ell)$ denotes the community to which node $\ell$ belongs. With the MCM $\bs M = n \bs I - \bs L$ defined as in Proposition~\ref{prop:prop-low-rank-unnormalized-sc}, we can rewrite this as \begin{equation}
    \operatorname{RatioCut}(Z_1, ..., Z_k) = n\operatorname{Tr}(\boldsymbol{S}^T\boldsymbol{S}) - \operatorname{Tr}(\boldsymbol{S}^T {\boldsymbol{M}} \boldsymbol{S})=nk - \operatorname{Tr}(\boldsymbol{S}^T {\boldsymbol{M}} \boldsymbol{S}).
\end{equation}
In particular, choosing $(Z_{1},\dots,Z_{k_c})$ minimizing $\operatorname{RatioCut}$ is equivalent to choosing $(Z_{1},\dots,Z_{k_c})$ maximizing $\operatorname{Tr}(\boldsymbol  S^{\top} \boldsymbol  M \boldsymbol  S)$. 

Take a spectral decomposition $\boldsymbol M= \boldsymbol V \boldsymbol \Lambda \boldsymbol V^{\top}$, with the entries of $\boldsymbol \Lambda$ positioned in descending order, and define $\boldsymbol N:= \boldsymbol V \boldsymbol \Lambda^{1/2}$. The $\ell$th row $\bs r_\ell$ of $\boldsymbol N$ equals the $\ell$th Euclidean node embedding derived by unnormalized spectral clustering when appropriately truncated, up to scalings from $\boldsymbol \Lambda^{1/2}$.\footnote{In fact, some versions of spectral clustering use $\bs N$ directly \cite{zhang_detecting_2020}.} We can rewrite our objective in terms of these node embeddings:  \begin{equation}
    \operatorname{Tr}(\boldsymbol S^{\top} \boldsymbol M \boldsymbol S)= \operatorname{Tr}(\boldsymbol{S}^T \boldsymbol{N} \ \boldsymbol{N}^T \boldsymbol{S}) = \sum_{j=1}^k \frac{1}{|Z_j|} \left\|\sum_{\ell \in Z_j} \boldsymbol{r}_\ell \right\|_2^2.
\end{equation}
In this light, we see that minimizing $\operatorname{RatioCut}$ is equivalent to a flavor of \textit{max-sum vector partitioning} applied to $d$-dimensional node embeddings derived from the leading eigenvectors of $\boldsymbol M$. Exactly solving max-sum vector partitioning problems is extremely expensive \cite{zhang_multiway_2015}; as such, various heuristics have been proposed \cite{zhang_multiway_2015, alpert_spectral_1995, wang2008vector}, including ones that are nearly equivalent to $k$-means \cite{zhang_multiway_2015}. Truncating to $k_c \leq p \leq d$ rows of $\boldsymbol N$ yields a max-sum vector partitioning problem in $\mathbb{R}^{p}$ whose objective approximates the $\operatorname{RatioCut}$ objective. When $p=k_c$ we recover an algorithm very similar to spectral clustering. When $p=d$ we are optimizing the $\operatorname{RatioCut}$ objective exactly rather than a low-rank approximation to it — if doing so were computationally feasible, this would theoretically give the best results. So choosing an embedding dimension $p=k_{e}$ exceeding the number of desired clusters $k_{c}$ is, in principle, an advantage rather than a detriment. In practice, due e.g. to noise considerations, dimensionality curses, and other challenges to which Euclidean clustering heuristics are generally susceptible in high-dimensional space, it is not recommended to let $k_{e}$ be arbitrarily high relative to $k_{c}$.  

We can utilize these ideas towards extending Algorithm~\ref{alg:geodesic-dcd} in the following manner. By definition, $k_{e}$ must be fixed over time in Algorithm~\ref{alg:geodesic-dcd}. However, nothing prevents $k_{c}=k_{c,i}=k_c(t_i)$ from varying based on $i \in [T]$, and — based on the discussion above — this should not incur any penalty compared to if we had chosen $k_{e}=k_{c,i}$ from the start, so long as $k_{c,i} \leq k_{e}$ and there exists earnest $k_{c}$-way community structure at time $i$. 

\paragraph{Automatically choosing the number of communities} How should $k_{c,i}$ be chosen? The easiest approach would be to evaluate some unsupervised partition score (e.g., modularity or the appropriate extension to other network modalities) with respect to different choices of $k_{c,i}$ and choose the maximizer. Of course, evolving community detection algorithms are in part motivated by the premise that 
the `correct’ community partition at a given time step may not always be the one maximizing the static modularity (e.g., due to measurement noise or stability considerations \cite{cazabet_dynamic_2017, rossetti_community_2018}). We therefore recommend incorporating some temporal awareness into the heuristic by filtering the scores for dynamic regularity, and reiterate our suggestion from Section~\ref{sec:experiments} to initialize the Euclidean clustering heuristic of $\mathscr{A}$ with clusters found at the previous time step. This is the approach taken in the following algorithm; potential alternatives are mentioned in Section~\ref{sec:variable-k-alternatives}.

\subsection{Algorithm and Evaluation}
\label{sec:variable-k-algorithm}
Informed by the previous discussion (appendix  \ref{sec:variable-k-motivation}), we provide an extension of Algorithm~\ref{alg:geodesic-dcd} to the case where the true number of communities earnestly varies over time, e.g., due to communities merging or splitting.

\begin{algorithm}[H]
\caption{Evolving Community Detection with Grassmann Geodesics — Variable $k_c$}
\label{alg:geodesic-dcd-variable-k}
\begin{algorithmic}[1]
\Require Graphs and associated snapshot times $\{ G_i, t_i \}_{i=1}^{T}$
\Require Estimates $k_{\text{min}}$ and $k_{\text{max}}$ of the minimum and maximum number of communities present at any time step
\Require Static spectral method $\mathscr{A}$ admitting an MCM (i.e., $\mathscr{A}$ formulated as in Algorithm~\ref{alg:static-scd-template})
\Require Unsupervised network partition benefit function (e.g., modularity for simple networks) 
\State Apply steps 1 and 2 of algorithm \ref{alg:geodesic-dcd} to fit a geodesic, using $k_{e}:=k_{\text{max}}$ as the embedding dimension
\State Instantiate $\boldsymbol H \in \mathbb{R}^{(k_{\max}-k_{\min}+1) \times T}$ 
\For{$k$ in $[k_{\min}, k_{\max}] \cap \mathbb{Z}$}
    \For{$i \in [T]$}
        \State \textbf{Euclidean clustering:} apply the Euclidean clustering heuristic of $\mathscr{A}$ (e.g., $k$-means) at snapshot $i$, storing the modularity $Q$ achieved by the resulting partition as $H_{ki} \leftarrow Q$. Recommended: initialize with results from time $i-1$ if applicable.
    \EndFor
    \State \textbf{Filtering:} Optionally, smoothen over the row $\boldsymbol H_{k, :}$ (our experiments convolve with a 1D Gaussian filter; median filtering is another straightforward option.) 
\EndFor
\For{$i \in [T]$}
    \State Define $k_{c,i}$ to be the choice of $k$ corresponding to the maximum element of the vector $\boldsymbol H_{:, i}$  
\EndFor
\Ensure Partitions $(Z_{1},\dots,Z_{k_{c, i}})_{i}$, $i \in [T]$, of each $G_i$ into communities
\end{algorithmic}
\end{algorithm}


Implementing Algorithm~\ref{alg:geodesic-dcd-variable-k} requires a benefit function suitable for the network modality. We remark that, at minimum, an extension of modularity exists in the literature for all network modalities considered in this text (signed \cite{he2023generalized}; overlapping \cite{nicosia2009extending}; directed \cite{leicht2008community}; multiview \cite{mucha_community_2010}; cocommunity \cite{barber2007modularity}; hierarchical 
\cite{reichardt2006statistical}). \new{We also remark that,  unless $k_{\text{max}}$ is quite large, it generally suffices to fix $k_{\text{min}}=2$ if no preferred choice is known. In practice, therefore, $k_{\text{max}}$ is usually the only required parameter. All experiments in this paper use the default $k_{\text{min}}=2$.}

Figure~\ref{fig:variable-k-results} shows representative results of Algorithm~\ref{alg:geodesic-dcd-variable-k} on a synthetic benchmark generated according to the methodology in \cite{cazabet2020evaluating}, wherein eight planted communities gradually merge into six over time. The extension of geodesic normalized spectral clustering to case of time-varying $k_c$ via Algorithm~\ref{alg:geodesic-dcd-variable-k} uniformly outperforms the benchmarks; indeed, it is the only method to recover the planted community structure at every snapshot.

\begin{figure}
    \centering
    \begin{subfigure}{\textwidth}
        \centering
        \includegraphics{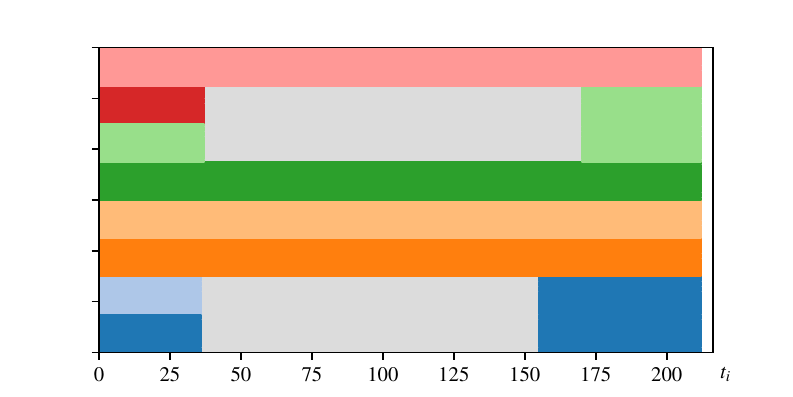}

    \end{subfigure}
    \vspace{1cm} 
    \begin{subfigure}{\textwidth}
        \centering
        \includegraphics {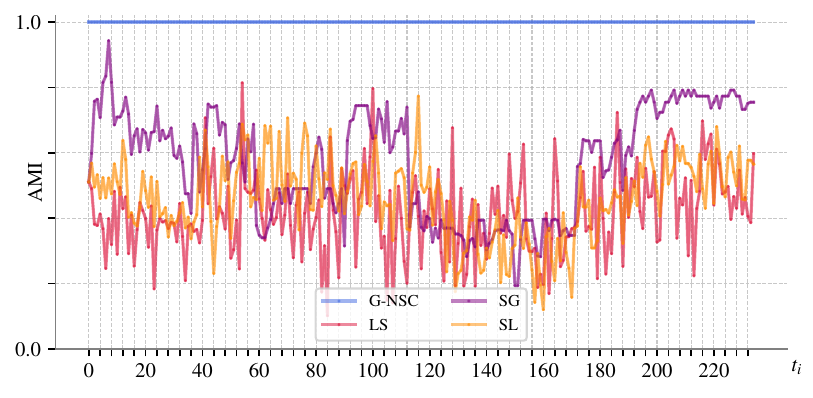}
    \end{subfigure}
    \caption{Evaluation of Algorithm~\ref{alg:geodesic-dcd-variable-k} on an evolving community detection benchmark generated using the \texttt{tnetwork} library \cite{cazabet_documentation_nodate} via the methodology from \cite{cazabet2020evaluating}. Initially, $8$ communities are present in the data. Over the course of $T=211$ snapshots, two pairs of communities merge, such that ultimately six communities remain. The other four communities remain stable throughout. LEFT: A longitudinal plot of the planted community structure over time. The red community gradually merges into the light green community, and the light blue community gradually merges into the dark blue community. Note that the benchmark, by convention, does not treat nodes as belonging to any community at all when they are transitioning; this is indicated in gray. RIGHT: \new{AMI} comparison over time for geodesic normalized spectral clustering (G-NSC, extended to detect varying $k_c$ via Algorithm~\ref{alg:geodesic-dcd-variable-k} with $k_{\min}=2$, $k_{\max}=10$) versus the  Label Smoothing (LS) approach of \cite{falkowski_mining_2006},  the Smoothed Louvain (SL) algorithm of \cite{aynaud_static_2010}, and the Graph Smoothing (SG) approach of \cite{guo_evolutionary_2014}. We remove from consideration the nodes currently lacking a true community membership when computing \new{AMI} at a given time step.}
    \label{fig:variable-k-results}
\end{figure}

\new{\paragraph{Evaluation on the temporal college football network} Figure~\ref{fig:real-college-football-results} shows the results of Algorithm~\ref{alg:geodesic-dcd-variable-k} applied to a temporal network inspired by the the popular American college football network of \cite{girvan2002community}, where an edge is placed between two football teams whenever a game is played between them over a ten-year span (2009-2019).\footnote{\new{The full dataset \newnew{is available at \href{https://github.com/jacobh140/century-of-college-football}{https://github.com/jacobh140/century-of-college-football}}.}} The ground truth community structure is provided by conference alignment, as in \cite{girvan2002community}. The community structure evolves as teams change their conference membership over time. The raw data is obtained via the API offered in \cite{cfbdataapi}, then segmented into one snapshot per month of each football season (September 2009, October 2009, November 2009, September 2010$\dots$) to obtain a time series of $T=33$ networks containing $d=118$ nodes each. A comparison of results can be seen in Figure~\ref{fig:real-college-football-results}. Geodesic normalized spectral clustering largely outperforms the benchmarks across time. We remark that the periodicity in some of the benchmarks' performance may be attributed to the fact that the density of in-conference games played is lower in September than in October and November. }

\begin{SCfigure}
    \includegraphics{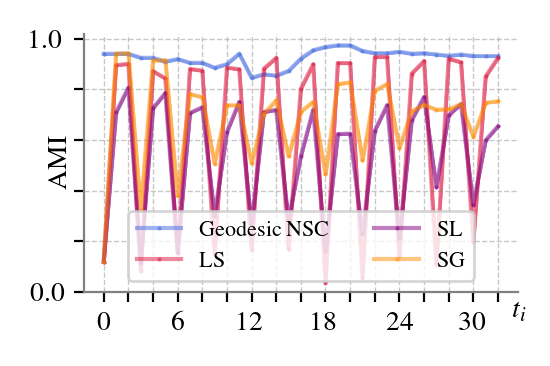}
    \caption{\new{Evaluation of Algorithm~\ref{alg:geodesic-dcd-variable-k} on a monthly extension of the college football network of ~\cite{girvan2002community}. Geodesic normalized spectral clustering (G-NSC, extended to detect varying $k_c$ via Algorithm~\ref{alg:geodesic-dcd-variable-k} with $k_{\min}=2$, $k_{\max}=10$) is compared against the  Label Smoothing (LS) approach of \cite{falkowski_mining_2006},  the Smoothed Louvain (SL) algorithm of \cite{aynaud_static_2010}, and the Graph Smoothing (SG) approach of \cite{guo_evolutionary_2014}, consistently outperforming the benchmarks over time.}} 
    \label{fig:real-college-football-results}
\end{SCfigure}

\subsection{Alternative Methods and Future Interest}
\label{sec:variable-k-alternatives}
The heuristic offered by Algorithm~\ref{alg:geodesic-dcd-variable-k}, though empirically effective and simple to implement, is nevertheless one of many options for extending Algorithm~\ref{alg:geodesic-dcd} to the setting where the number of communities changes over time. One alternative to the `benefit function + sweep $k_c$' approach detailed previously would be to employ one of many methods for automated model selection in $k$-means if applicable, using e.g. the GAP statistic \cite{tibshirani2001estimating} or silhouette scores \cite{shahapure2020cluster, rousseeuw1987silhouettes} either as a benefit function for Algorithm~\ref{alg:geodesic-dcd-variable-k} or standalone, or using $X$-means \cite{pelleg2000x}. We also note that the geodesic generalization of hierarchical community detection described in Section~\ref{sec:hierarchical-networks}, put in conjunction with Algorithm~\ref{alg:geodesic-dcd-variable-k} by using the latter to estimate the size of the initial partition, yields a dynamic algorithm for detecting hierarchical communities that is hyperparameter-free. It is of future interest to seek a variant of Algorithm~\ref{alg:geodesic-dcd} which circumvents the requirement of a temporally fixed embedding dimension $k_e$ altogether.

\section{Remarks on Scaling}
\label{sec:additional-experiments-appendix}

\new{The worst case per-iteration time complexity of the geodesic-fitting step in Algorithm~\ref{alg:geodesic-dcd} and Algorithm~\ref{alg:geodesic-dcd-variable-k} is $O(Td^2k_e)$, where $k_e$ is the embedding dimension. The other potentially expensive step is the multiple applications of a Euclidean clustering heuristic. In the case of $k$-means, various techniques have been devised to increase efficiency \cite{bahmani2012scalable, shindler2011fast, hamerly2010making}, though none of our experiments make use of these.}

\new{In this section, we assess the scaling of Algorithm~\ref{alg:geodesic-dcd-variable-k} with respect to $d$ and $T$ on the dynamic stochastic block model introduced in Section~\ref{sec:experiments}. We implement Algorithm~\ref{alg:geodesic-dcd-variable-k} for normalized spectral clustering in Python on a 2019 MacBook Pro (2.3 GHz 8-Core Intel Core i9 Processor, 16GB  RAM), using sparse matrices and the $k$-means$++$ implementation found in \texttt{scikit-learn}. The other algorithms are compared using their implementations in the \texttt{tnetwork} \texttt{Python} library. We find that, for the parameter settings, hardware, and implementations considered, the geodesic recovered or nearly recovered the true community structure, outperforming the benchmarks while running about an order of magnitude ($10.47$x on average) faster.}

\begin{table}[h]
\centering
\small
\begin{tabular}{|@{\hspace{2pt}}c@{\hspace{2pt}}|@{\hspace{2pt}}c@{\hspace{2pt}}|@{\hspace{2pt}}c@{\hspace{2pt}}|@{\hspace{2pt}}c@{\hspace{2pt}}|@{\hspace{2pt}}c@{\hspace{2pt}}|@{\hspace{2pt}}c@{\hspace{2pt}}|}
\hline
$d$ & $T$ & GNSC & LS & SL & SG \\ \hline
$10^3$ & $10^2$ & $\mathbf{1.0 \pm 0.0}$ \textbf{(36s)} & $0.65 \pm 0.32$ (481s) & $0.86 \pm 0.02$ (144s) & $0.80 \pm 0.03$ (349s) \\ \hline
$10^3$ & $10^3$ & $\mathbf{1.0 \pm 0.0}$ \textbf{(387s)} & $0.65 \pm 0.33$ (5893s) & $0.87 \pm 0.02$ (2172s) & $0.79 \pm 0.01$ (5445s) \\ \hline
$10^4$ & $10^2$ & $\mathbf{0.99 \pm 5 \times 10^{-5}}$ \textbf{(641s)} & {N/A} & $0.71 \pm 0.11$ (7266s) & {N/A} \\ \hline
\end{tabular}
\caption{\new{Comparison of mean AMI and empirical runtime on the dynamic stochastic block model ($p_{\text{out}}=200d^{-1}$, $p_{\text{in}}=300d^{-1}$, $p_{\text{switch}}=(10T)^{-1}$, $k=5$) as size and longitude vary. The methods compared are geodesic normalized spectral clustering (G-NSC, ours), label smoothing (LS, \cite{falkowski_mining_2006}), smoothed Louvain (SL, \cite{aynaud_static_2010}), and smoothed graph (SG, \cite{guo_evolutionary_2014}), as in Section~\ref{sec:experiments}. A `N/A' designation is given wherever a method timed out during a simulation.} }
\label{tab:scaling-experiments}
\end{table}

\end{document}